\documentclass[runningheads]{llncs}

\usepackage{fullpage}
\usepackage[utf8]{inputenc}
\usepackage{amsmath}
\usepackage{hyperref,cleveref}
\usepackage{algorithm}
\usepackage[noend]{algpseudocode}
\usepackage{amsfonts}
\usepackage{amssymb}
\usepackage{mathrsfs}
\usepackage{mathtools}
\usepackage{xcolor}
\usepackage{wrapfig}
\usepackage{breqn}
\usepackage{cite}
\usepackage{thm-restate}
\usepackage{tikz}
\usepackage{pgfplots}

\usepackage{amsthm}
\usepackage{orcidlink}
\usepackage{subcaption}

\allowdisplaybreaks
\DeclarePairedDelimiter{\norm}{\lVert}{\rVert}
\DeclarePairedDelimiter{\abs}{\lvert}{\rvert}
\DeclarePairedDelimiter{\parens}{(}{)}
\DeclarePairedDelimiter{\bracks}{[}{]}
\DeclarePairedDelimiter{\braces}{\{}{\}}
\DeclarePairedDelimiter{\tup}{\langle}{\rangle}
\DeclareMathOperator*{\argmin}{arg\,min}
\DeclareMathOperator*{\argmax}{arg\,max}
\DeclareMathOperator*{\E}{\mathbb{E}}
\DeclareMathOperator*{\Var}{\mathbb{V}}
\DeclareMathOperator*{\Prob}{\mathbb{P}}
\newcommand{\R}{\mathbb{R}}
\newcommand{\dd}{\mathcal{D}}
\newcommand{\loss}{\ell}
\newcommand{\diameter}[2][{}]{\mathrm{diam}^{#1}\parens*{#2}}

\newcommand{\cluster}{\mathsf{cluster}}

\newcommand{\epsconverge}{\epsilon}
\newcommand{\epsmodel}{\delta}
\newcommand{\fp}{f}

\newcommand{\fopt}{f^*}
\newcommand{\thresh}{n_b}
\newcommand{\explain}[1]{{\color{blue}\text{#1}}}

\newcommand{\remove}[1]{}

\begin{document}

\title{Asynchronous Fully-Decentralized SGD in the Cluster-Based Model\thanks{
    This work was supported by Pazy grant 226/20 and 
    the Israel Science Foundation grants 380/18 and 22/1425.}}
\author{Hagit Attiya \orcidlink{0000-0002-8017-6457} \and Noa Schiller \orcidlink{0009-0007-0285-6194}}
\institute{Department of Computer Science, Technion, Israel \\
\email{\{hagit,noa.schiller\}@cs.technion.ac.il}}

\authorrunning{Hagit Attiya \and Noa Schiller}

\maketitle

\begin{abstract}
This paper presents fault-tolerant asynchronous \emph{Stochastic 
Gradient Descent} (\emph{SGD}) algorithms.
SGD is widely used for approximating the minimum of a cost function $Q$,
a core part of optimization and learning algorithms. 
Our algorithms are designed for the \emph{cluster-based} model,
which combines message-passing and shared-memory communication layers.
Processes may fail by \emph{crashing}, and the algorithm inside 
each cluster is \emph{wait-free}, using only reads and writes.

For a \emph{strongly convex} $Q$,
our algorithm \emph{can withstand partitions of the system}.
It provides convergence rate that is the maximal distributed 
acceleration over the optimal convergence rate of \emph{sequential} SGD.

For arbitrary smooth functions, the convergence rate has an additional term that
depends on the maximal difference between the parameters at the same iteration.
(This holds under standard assumptions on $Q$.)
In this case, the algorithm obtains the same convergence rate as sequential SGD, 
up to a logarithmic factor. 
This is achieved by using, at each iteration, a \emph{multidimensional 
approximate agreement} algorithm, tailored for the cluster-based model. 

The general algorithm communicates with nonfaulty processes 
belonging to clusters that include a majority of all processes.
We prove that this condition is necessary when optimizing some 
non-convex functions.
\end{abstract} 


\section{Introduction}

An \emph{optimization} problem attempts to minimize the value of 
a \emph{cost function} $Q : \R^d \rightarrow \R$,
that is, find $x^*\in \argmin_{x\in \R^d} Q(x)$.
Among their many uses, 
optimization problems play a key role in machine and deep learning~\cite{lecun2015deep},
often using \emph{stochastic gradient descent} (SGD).
SGD~\cite{sgd} repeatedly applies the update rule
$\mathbf{x}_{t+1} = \mathbf{x}_t - \eta_t G(\mathbf{x}_t,z_t)$,
in each iteration $t$.
The arguments for this rule are the \emph{learning parameter} $\mathbf{x}_t$,
the \emph{learning rate} $\eta_t$,
and a random sample $z_t$ from \emph{data distribution} $\dd$;
$G(\mathbf{x}_t,z_t)$ computes the \emph{stochastic gradient} 
of $\mathbf{x}_t$ and $z_t$, 
which is an \emph{unbiased estimator} of the true gradient 
$\nabla Q(\mathbf{x}_t)$.
Intuitively, the gradient points to the direction of the steepest 
slope at that point, and its opposite direction gives the biggest 
(local) decrease in function value. 
When the function $Q$ is strongly convex, 
$\mathbf{x}_t$ will converge to the unique minimum of $Q$~\cite{bottou2018optimization};
otherwise, it will converge to a point with zero gradient~\cite{GhadimiL13a}.


In learning applications, SGD is applied to a function $Q$ of high dimension $d$, 
using many stochastic gradients~\cite{NIPS2012_6aca9700}.
The convergence of the basic SGD algorithm can be improved by 
\emph{mini-batch SGD},
which computes $b$ stochastic gradients using $b$ samples,
drawn uniformly at random from $\dd$. 
The average of these $b$ gradients has variance that is a factor of $b$ smaller
than $\sigma^2$, the variance of a single stochastic gradient, 
implying a linear reduction in the number of iterations for convergence.
Since gradients are computed independently, 
(mini-batch) SGD is a prime target for large-scale distributed and 
parallel computing.
Mini-batch SGD provides a baseline for measuring 
performance in the distributed setting.

In an iteration of a typical distributed SGD algorithm, 
a \emph{worker} performs some local computation and then, 
all computed values are aggregated to collectively compute 
parameters for the next iteration~\cite{Ben-NunH19}.
This can be done in a \emph{centralized} manner, 
where a \emph{parameter server} aggregates all the computation 
done by the workers (e.g.,~\cite{AlistarhSK18,RechtRWN11,
NIPS2015_452bf208,NEURIPS2018_a07c2f3b,NIPS2011_f0e52b27,El-MhamdiGGHR20,AfekGPS21}), 
or in a \emph{decentralized} manner,
where each worker holds a copy of the parameters (e.g.,~\cite{wagmaSGD,
li2020taming,NIPS2017_f7552665,pmlr-v80-lian18a,
CollaborativeLearningNeurIPS}).
A straightforward, synchronized implementation of (mini-batch) SGD 
requires locks or barriers to ensure that workers proceed in lock-step, 
thereby harming the performance.

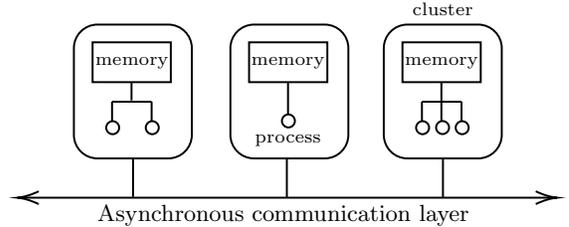
\begin{wrapfigure}{R}{0.46\textwidth}
\centering
\resizebox{0.45\textwidth}{!}{\tikzset{every picture/.style={line width=0.75pt}} 

\begin{tikzpicture}[x=0.75pt,y=0.75pt,yscale=-1,xscale=1]

\draw    (53.4,109.88) -- (326.6,109.88) ;
\draw [shift={(328.6,109.88)}, rotate = 180] [color={rgb, 255:red, 0; green, 0; blue, 0 }  ][line width=0.75]    (10.93,-3.29) .. controls (6.95,-1.4) and (3.31,-0.3) .. (0,0) .. controls (3.31,0.3) and (6.95,1.4) .. (10.93,3.29)   ;
\draw [shift={(51.4,109.88)}, rotate = 0] [color={rgb, 255:red, 0; green, 0; blue, 0 }  ][line width=0.75]    (10.93,-3.29) .. controls (6.95,-1.4) and (3.31,-0.3) .. (0,0) .. controls (3.31,0.3) and (6.95,1.4) .. (10.93,3.29)   ;
\draw   (80.4,33) .. controls (80.4,26.31) and (85.83,20.88) .. (92.52,20.88) -- (128.88,20.88) .. controls (135.57,20.88) and (141,26.31) .. (141,33) -- (141,77.56) .. controls (141,84.26) and (135.57,89.68) .. (128.88,89.68) -- (92.52,89.68) .. controls (85.83,89.68) and (80.4,84.26) .. (80.4,77.56) -- cycle ;
\draw   (90,30) -- (130.2,30) -- (130.2,50.48) -- (90,50.48) -- cycle ;
\draw    (110,50.68) -- (110,59.88) ;
\draw    (100,60.55) -- (120.33,60.55) ;
\draw    (100,60.55) -- (100,70.55) ;
\draw   (97,73.88) .. controls (97,72.04) and (98.49,70.55) .. (100.33,70.55) .. controls (102.17,70.55) and (103.67,72.04) .. (103.67,73.88) .. controls (103.67,75.72) and (102.17,77.22) .. (100.33,77.22) .. controls (98.49,77.22) and (97,75.72) .. (97,73.88) -- cycle ;
\draw    (120.33,60.55) -- (120.33,70.55) ;
\draw   (117.33,73.88) .. controls (117.33,72.04) and (118.83,70.55) .. (120.67,70.55) .. controls (122.51,70.55) and (124,72.04) .. (124,73.88) .. controls (124,75.72) and (122.51,77.22) .. (120.67,77.22) .. controls (118.83,77.22) and (117.33,75.72) .. (117.33,73.88) -- cycle ;
\draw   (160.8,33) .. controls (160.8,26.31) and (166.23,20.88) .. (172.92,20.88) -- (209.28,20.88) .. controls (215.97,20.88) and (221.4,26.31) .. (221.4,33) -- (221.4,77.56) .. controls (221.4,84.26) and (215.97,89.68) .. (209.28,89.68) -- (172.92,89.68) .. controls (166.23,89.68) and (160.8,84.26) .. (160.8,77.56) -- cycle ;
\draw   (170.4,30) -- (210.6,30) -- (210.6,50.48) -- (170.4,50.48) -- cycle ;
\draw    (190.4,50.68) -- (190.4,67.08) ;
\draw   (187.27,70.42) .. controls (187.27,68.58) and (188.76,67.08) .. (190.6,67.08) .. controls (192.44,67.08) and (193.93,68.58) .. (193.93,70.42) .. controls (193.93,72.26) and (192.44,73.75) .. (190.6,73.75) .. controls (188.76,73.75) and (187.27,72.26) .. (187.27,70.42) -- cycle ;
\draw   (239.6,33) .. controls (239.6,26.31) and (245.03,20.88) .. (251.72,20.88) -- (288.08,20.88) .. controls (294.77,20.88) and (300.2,26.31) .. (300.2,33) -- (300.2,77.56) .. controls (300.2,84.26) and (294.77,89.68) .. (288.08,89.68) -- (251.72,89.68) .. controls (245.03,89.68) and (239.6,84.26) .. (239.6,77.56) -- cycle ;
\draw   (249.2,30) -- (289.4,30) -- (289.4,50.48) -- (249.2,50.48) -- cycle ;
\draw    (269.2,50.68) -- (269.2,70.48) ;
\draw    (259.2,60.55) -- (279.53,60.55) ;
\draw    (259.2,60.55) -- (259.2,70.55) ;
\draw   (256.2,73.88) .. controls (256.2,72.04) and (257.69,70.55) .. (259.53,70.55) .. controls (261.37,70.55) and (262.87,72.04) .. (262.87,73.88) .. controls (262.87,75.72) and (261.37,77.22) .. (259.53,77.22) .. controls (257.69,77.22) and (256.2,75.72) .. (256.2,73.88) -- cycle ;
\draw    (279.53,60.55) -- (279.53,70.55) ;
\draw   (276.53,73.88) .. controls (276.53,72.04) and (278.03,70.55) .. (279.87,70.55) .. controls (281.71,70.55) and (283.2,72.04) .. (283.2,73.88) .. controls (283.2,75.72) and (281.71,77.22) .. (279.87,77.22) .. controls (278.03,77.22) and (276.53,75.72) .. (276.53,73.88) -- cycle ;
\draw    (110.8,89.28) -- (110.8,109.48) ;
\draw    (189.8,89.68) -- (189.8,109.88) ;
\draw    (270,89.68) -- (270,109.88) ;
\draw   (266.47,73.82) .. controls (266.47,71.98) and (267.96,70.48) .. (269.8,70.48) .. controls (271.64,70.48) and (273.13,71.98) .. (273.13,73.82) .. controls (273.13,75.66) and (271.64,77.15) .. (269.8,77.15) .. controls (267.96,77.15) and (266.47,75.66) .. (266.47,73.82) -- cycle ;

\draw (89.9,36) node [anchor=north west][inner sep=0.75pt]  [font=\scriptsize] [align=left] {{memory}};
\draw (170.1,36) node [anchor=north west][inner sep=0.75pt]  [font=\scriptsize] [align=left] {{memory}};
\draw (249,36) node [anchor=north west][inner sep=0.75pt]  [font=\scriptsize] [align=left] {{memory}};
\draw (252.8,8) node [anchor=north west][inner sep=0.75pt]  [font=\scriptsize] [align=left] {cluster};
\draw (172,76) node [anchor=north west][inner sep=0.75pt]  [font=\scriptsize] [align=left] {process};
\draw (91,112) node [anchor=north west][inner sep=0.75pt]  [font=\footnotesize] [align=left] {Asynchronous communication layer};

\end{tikzpicture}}
\caption{Cluster-based model}
\label{fig:cluster-based}
\end{wrapfigure}

This paper considers completely decentralized and {fully-}asynchronous SGD 
in a \emph{cluster-based model}~\cite{RaynalC19} that combines 
both shared memory and message passing (Figure~\ref{fig:cluster-based}):
processes are partitioned to disjoint clusters,
each sharing a memory space, accessed with reads and writes; 
additionally, all processes can communicate by message passing.
Processes may fail by \emph{crashing}, that is, stopping to take steps.
This model is interesting from a practical point-of-view,
as it captures several system architectures, 
for example \emph{high-performance computing} systems~\cite{HPC}.

\emph{Our first main contribution} 
shows that when the function $Q$ is strongly convex, 
a \emph{simple asynchronous} algorithm that collects $\thresh$ messages in each iteration, 
matches the convergence rate of a sequential mini-batch SGD algorithm 
(for strongly-convex functions)
with batch size $\thresh$~\cite{bottou2018optimization,AgarwalBRW12}.
Our analysis of this algorithm is relatively simple, 
and leverages the strong convexity of $Q$ to prove convergence, despite 
the fact that each process holds a local copy of the learning parameters.
Specifically, we prove (Theorem~\ref{thm:strongly-convex}) 
that if $Q$ is a smooth strongly-convex function,
then the convergence rate of the algorithm
after $T$ iterations with parameter $\thresh$
is ${O} \parens*{1/{\thresh T}}$.
Progress is ensured as long as {at least $\thresh$ processes do not fail}.

The algorithm for strongly-convex functions does not rely on communication
within clusters and it applies to general message-passing systems.
Even more importantly, 
{the algorithm works even if the system partitions,
that is, the computation proceeds within disjoint sets of processes}.
Roughly, since $Q$ is strongly convex, it has a single minimum,
and thus, all processes will independently converge to this minimum.

\emph{Our second main contribution} (Theorem~\ref{thm:thmlb})
shows that in general, when $Q$ is not strongly convex, 
any SGD algorithm requires that 
the set of non-failed processes \emph{represents} a majority of the processes, 
when a process represents all the processes in its cluster.
In the special case where communication is only through message passing, 
this reduces to {requiring that the number of processes $n$ is larger than $2f$}.


\emph{Our third main contribution} is a general SGD algorithm,
under the same assumption that the set of non-failed processes represents 
a majority of the processes.
It has a weaker convergence guarantee, relative to the baseline,
with an additional term $\Delta$, depending on the difference 
between the learning parameters of the different processes 
during the algorithm execution. 
We show 
that if $Q$ is a smooth function, 
then the convergence rate of the algorithm after $T$ iterations with parameter $\thresh \leq n-f$ is $\mathcal{O} \parens*{1/{\sqrt{\thresh T}} + T \Delta}$.
%
%
The first term 
matches the convergence rate achieved in standard analysis 
for non-convex objectives~\cite{GhadimiL13a}.

Unlike the strongly convex case, where the difference between the 
learning parameters is intrinsically bounded, here we bound $\Delta$ 
using \emph{multidimensional approximate agreement} (\emph{MDAA})~\cite{MendesH13}.
In MDAA, processes start with inputs in $\R^d$, 
and the outputs of nonfaulty processes should be ``close together''
and in the \emph{convex hull} of their inputs.
We use a shared-memory adaptation of~\cite{fggerFastApproximate} 
to bound the difference between the values sent from the same cluster.
By ensuring every pair of processes communicate with a representative 
process (not necessarily the same process) from at least one 
common cluster, we ensure good contraction at each iteration.
MDAA encapsulates the use of the shared memory at each cluster,
as well as the fault-tolerance required from the algorithm. 
This algorithm is interesting by itself, 
beyond its application in distributed learning.

Each MDAA iteration contains several \emph{communication rounds},
where each process sends a message and receives responses 
representing $n - f$ processes.
We prove that our general SGD algorithm can match the convergence rate 
of the sequential algorithm, \emph{up to a logarithmic blowup}
in the number of communication rounds.
Specifically (Theorem~\ref{thm:non-convex-small-error}), 
for a smooth function $Q$, the convergence rate of the algorithm 
after $R$ communication rounds 
is $ \widetilde{\mathcal{O}} \parens*{1/{\sqrt{\thresh R}}}$.


In our algorithms, each process serves as both a computation and a communication thread, 
executing the computation and sending its result to the other processes. 
Within the cluster, the algorithm is wait-free and uses only reads and writes 
to the cluster's shared memory;
no locks or barriers are used. 
Since processes operate in an independent manner, 
the algorithms achieve a speedup in the total number 
of processes and not the number of clusters.
Our algorithms also improve the resilience relative to a pure 
message passing model, 
since a process can \emph{represent} its cluster~\cite{AttiyaKS20}.

\section{Related Work}
\label{sec:related}

In the decentralized setting, \emph{synchronous} SGD can use \emph{allreduce}~\cite{mpi40} 
to exchange the stochastic gradients among all workers in each iteration.
To reduce synchronization cost of allreduce, 
it was suggested to partially reduce information within non overlapping groups of processes, 
in each iteration~\cite{wagmaSGD}.
However, since the reduce operation is activated by early workers, 
this could lead to averaging \emph{stale} gradients with more advanced ones. 
This can be bounded by performing a full allreduce every few iterations.
It is also possible to collect non-stale gradients from at least half of the 
workers~\cite{li2020taming}, ensuring convergence when staleness is bounded.
Our work has some similarities to these two papers. 
In a sense, we also relax the allreduce operation, where all processes 
send messages to each other but do not wait to receive from all of them. 
Similarly to~\cite{li2020taming}, our algorithm also waits for enough 
workers to move on to the next iteration. 
However, we ignore the stale gradients from previous iterations, 
allowing us to make no assumption on the staleness of the gradients.
This also means that workers will ignore these gradients
when they eventually receive them.

\emph{Elastic consistency}~\cite{Alistarh2021-ElasticConsistency} is a 
framework that assumes that the difference between the parameter 
used to compute the stochastic gradient by a process 
and the actual global parameter is bounded.
SGD converges under this assumption, 
for both convex and non-convex objective functions.  
This framework is agnostic to the system model and can be applied to 
several existing frameworks and algorithms.
They prove that several popular frameworks obey the above assumption. 
One example is in shared-memory~\cite{AlistarhSK18,RechtRWN11}, 
where processes access the same learning parameter stored in memory and 
update it, each coordinate at a time, using fetch\&add.
Another example is with message-passing~\cite{NIPS2015_452bf208,NIPS2011_f0e52b27}, 
where the parameter server may receive stale gradients, 
i.e., gradients computed using old parameters.
Both cases assume \emph{bounded asynchrony}, 
with a maximum delay $\tau$ on the staleness of gradients. 
Our algorithms are completely asynchronous, 
and do not assume any bound on the difference between the iterations 
different processes are in at any point in the execution.
We achieve this by ignoring gradients from past iterations, 
allowing us to explicitly bound the convergence, 
without bounding the maximum delay or relying on elastic consistency. 

Another line of work in the decentralized setting is when 
communication is dictated by a graph,
either synchronously~\cite{NIPS2017_f7552665}, 
or asynchronously~\cite{pmlr-v80-lian18a}.
This is {related} to \emph{gossip SGD}~\cite{gossipingSGD}, 
where the communication graph is random. 
Choosing the right (possibly random and time varying) communication graph 
allows information propagation from each process to all the others. 
This allows to solve \emph{average consensus}~\cite{Boyd,1333204,Tsitsiklis84}, 
a problem widely studied in the literature on control theory. 
Average consensus can be used to solve distributed optimization 
and ensure workers agree on the final estimate~\cite{nedic2009}. 
We do not make any assumptions on the communication graph, 
except that workers receive enough messages in each iteration. 
Therefore, information is not guaranteed to propagate during an execution 
of the algorithm and we must explicitly bound the difference in estimates 
among the workers.

 

\emph{Approximate agreement} was originally defined over the reals~\cite{DolevLPSW86}.
\emph{Multidimensional approximate agreement} was defined for asynchronous 
systems with malicious failures~\cite{MendesH13}, requiring 
that the outputs of nonfaulty processes' are  close together and
in the \emph{convex-hull} of their inputs.
They prove that the optimal fault-tolerance for this problem is $n>f(d+2)$, 
for inputs in $\R^d$. 
Algorithms for multidimensional approximate agreement, assuming both crash and malicious failures, 
with constant contraction rate are {known}~\cite{fggerFastApproximate}.

The lower bound of $n>f(d+2)$ can be circumvented by using \emph{averaging 
agreement}~\cite{CollaborativeLearningNeurIPS}, where convexity 
is replaced by the requirement that the Euclidean difference between 
the average of the nonfaulty inputs and outputs should be proportional 
to the initial Euclidean distance between the inputs.
Using averaging agreement allows to prove the convergence of the \emph{average} 
of the outputs of nonfaulty nodes~\cite{CollaborativeLearningNeurIPS}.
The lower bound of $n>f(d+2)$~\cite{MendesH13} does not hold with
crash failures, and we solve asynchronous MDAA assuming only $n>2f$. 
Using convexity of MDAA, we prove that 
the learning parameters at each process converge \emph{individually}, 
and not just the average parameters over all processes.
\emph{Byzantine  optimization} was introduced in \cite{SuVaidya}.
It was shown that \emph{$2f$-redundancy} is a necessary and 
sufficient condition for $f$-resilient Byzantine deterministic 
optimization, both exact~\cite{Redundancy-PODC20}
and approximate~\cite{ApproxByzOptimization-PODC21}. 
The lower bound of $n>2f$ relies on injecting malicious values, 
and does not capture a partitioning of the system.

In the heterogeneous case, where each process has a different cost function, 
averaging agreement is equivalent to asynchronous decentralized Byzantine 
learning~\cite{CollaborativeLearningNeurIPS}. 
This paper shows two algorithms, one that requires $n \geq 6f +1$, 
achieving the best-possible asymptotic averaging constant,
and another that requires only $n \geq 3f+1$, 
by using improved techniques for approximate agreement~\cite{AbrahamAD04}.
They also prove a lower bound of $n>3f$ for averaging agreement, 
implying the same lower bound for Byzantine learning.
%
Our algorithms are similar to the homogeneous-cost 
algorithm of~\cite{CollaborativeLearningNeurIPS},
but their algorithm tolerates malicious failures while 
our algorithm tolerates only crash failures.
Our algorithms are tailored to the cluster-based model, 
exploiting shared-memory communication, 
while their algorithms rely only on message passing.
Our analysis also accounts for the extra communication rounds 
of the averaging agreement protocol, which is ignored in their analysis. 
While the algorithm of~\cite{CollaborativeLearningNeurIPS} could also 
tolerate crash failures, it is not optimized for this case and 
does not achieve a speedup {proportional to} the number of workers.
Our algorithms improve the convergence rate as a function 
of the number of participating processes,
a feature that is lacking in many Byzantine-resilient SGD 
algorithms (see~\cite{MaurerOPODIS21}).
An exception is the algorithm of~\cite{NEURIPS2018_a07c2f3b}, 
which achieves speedup in the number of processes, 
but this algorithm is synchronous and centralized.


A hybrid setting, with a set of parameter servers and a set of workers
is considered in~\cite{El-MhamdiGGHR20}. 
They assume that the number of workers is $\geq 3f_w + 1$, 
where at most $f_w$ workers can be malicious,
and the number of parameter servers is $\geq 3f_{ps} + 2$, 
where at most $f_{ps}$ parameter servers can be malicious.
They make a non-standard assumption that any communication pattern 
between non-malicious servers eventually holds with probability 1.
They consider smooth non-convex cost functions, 
but do not bound the convergence rate of their algorithm 
and only show eventual convergence. 

\emph{Federated learning}~\cite{KairouzMABBBBCC21} considers the 
centralized setting with synchronous iterations; 
however, not all the workers participate in each learning iteration.
A prominent algorithm is \emph{FedAvg}~\cite{McMahanMRHA17},
where each worker does one or more SGD iterations and sends its updated 
parameters to the parameter server,
which simply averages them to compute the parameters for the next iteration. 
It was shown~\cite{karimireddy2021scaffold} that non-uniform distribution 
of data across the workers can hinder the convergence of FedAvg, 
by introducing drift in the updates sent by the workers to the server. 
SCAFFOLD~\cite{karimireddy2021scaffold} adds a correction term 
to mitigate the drift introduced by the workers. 
Despite the differences between our model and the federated learning 
framework, our proofs utilize some techniques from~\cite{karimireddy2021scaffold}.

\section{Preliminaries}
\label{sec:model}

\subsection{Model of Computation}

There are $n$ \emph{processes}, $1,\dots,n$, 
which are partitioned into $m\leq n$ disjoint \emph{clusters}, 
$P_1,...,P_m$. 
Formally, $P_i\subseteq \{1,\dots,n\}$, 
$P_i\cap P_j = \emptyset$ for every $1\leq i<j\leq m$, 
and $\bigcup_{i=1}^m P_i = \{1,\dots,n\}$.
Given a process $i$, $\cluster(i)$ is its cluster, i.e., 
$\cluster(i) = P_j$ such that $i\in P_j$. 
We slightly abuse notation, and for a set $V\subseteq [n]$, 
denote $\cluster(V) = \cup_{i\in V} \cluster(i)$.
Processes may \emph{crash} and stop taking steps. 
A process is \emph{nonfaulty} if it never stops taking steps. 
For simplicity, a process keeps taking (empty) steps even after
it completes the algorithm.
A \emph{shared memory} is associated with each cluster,
and it is accessed with read and write operations, 
only by the processes in the cluster. 
In addition, each process can send messages to each other process, 
using an asynchronous, reliable communication link. 
This means that messages between nonfaulty processes are eventually 
delivered, but {they can be arbitrarily delayed}.
Figure~\ref{fig:cluster-based} shows $n=6$ processes organized in $m=3$ clusters:
$P_1 = \{1,2\}, P_2 = \{3\}$ and $P_3 = \{4,5,6\}$.
The maximal number of processes that can crash is denoted $\fp$, 
$1 \leq \fp < n$.

A process $i$ \emph{represents} all the processes in $\cluster(i)$.
Let $\fopt$ be the maximal integer such that any set $P\subseteq [n]$ 
of size $n-\fopt$ represents a majority of the processes, i.e., 
$|{\cluster(P)}| \geq \lfloor n/2 \rfloor + 1$.
Intuitively, this is the exact number of failures a system can withstand 
without \emph{partitioning}, i.e., the situation where 
two disjoint sets of processes run without communication.
In the special case of a pure message-passing system (with singleton 
clusters) requiring $\fp \leq \fopt$ amounts to $f<n/2$.
Note that $\fopt \geq \lfloor (n-1)/2 \rfloor$, and we have:

\begin{lemma}[\cite{AttiyaKS20}]
\label{lemm:fopt-comm}
    If $f > \fopt$, then there are two sets of processes $P, Q \subseteq [n]$, each of size $n-f$, such that $\cluster(P) \cap \cluster(Q) = \emptyset$. In other words, any process in $P$ and process in $Q$ cannot communicate using shared memory.
\end{lemma}

\begin{lemma}[\cite{AttiyaKS20}]
\label{lemm:fopt}
    If $\fp \leq \fopt$, then any two sets $P,Q\subseteq [n]$,
    each representing $n - \fp$ processes, must include a process from the same cluster.
\end{lemma}

A \emph{configuration} $C$ consists of the local state of each process, 
pending messages that were not received yet 
and the shared memory state of each cluster. 
An \emph{event} is either a delivery of some message by process $i$ 
or some operation on its cluster shared memory. 
A \emph{step} consists of some local computation, 
possibly a \emph{coin flip} {(i.e., generating some local randomness)}, and a single event.
By applying a step performed by process $i$ to configuration $C$, 
we obtain a new configuration with a new local state for process $i$, 
possibly removing or adding messages from the pending messages buffer 
and the updated shared memory state of $i$'s cluster.

Given a configuration $C$, for every process there is a fixed 
probability for every step it can take from $C$.
An \emph{execution tree} $\mathcal{T}$
is a directed weighted tree where each node is a configuration 
and the edges are all the possible steps that can be taken 
from this configuration. 
The weight on each edge is exactly the probability for the step 
to be taken from the parent configuration. 
The root of the execution tree is an \emph{initial configuration}.
Any path in the execution tree, beginning from the root, 
defines a legal \emph{execution}. 
The probability over the execution tree for an execution to occur is the product of weights along the path that defines the execution. The execution tree induces a well-defined probability
space, which is not explicitly defined.


An execution tree $\mathcal{T}$ is \emph{valid} if all the outgoing edges from nodes in the same level of the tree are all the possible steps of the same process. In other words, the process \emph{schedule} of each execution in the tree, which is the sequence of processes in the order that they take steps, are identical.
In addition, each execution is infinite and 
consists of at most $f$ failures. Note that this induces scheduling of a \emph{weak static adversary} that cannot observe the local states of the processes and has to choose the schedule in advance.



\subsection{Stochastic Gradient Descent}

The Euclidean norm of a vector $\mathbf{x}=(x_1,...,x_d)\in \R^d$
is $\norm*{\mathbf{x}}_2 \triangleq \sqrt{\sum_{i=1}^d\abs*{x_i}^2}$;
we use the standard notation, 
$\norm*{\mathbf{x}}_2^2 \triangleq \parens{\norm*{\mathbf{x}}_2}^2 = \sum_{i=1}^d\abs*{x_i}^2$.
The \emph{inner product} of $\mathbf{x},\mathbf{y}\in \R^d$ is $\tup{\mathbf{x},\mathbf{y}} \triangleq \mathbf{x}^T\mathbf{y} = \sum_{i=1}^d x_i y_i$.
The \emph{variance} of a random vector $\mathbf{x}$ is $\Var\bracks*{\mathbf{x}} \triangleq \E \bracks*{\norm*{\mathbf{x} - \E \bracks*{\mathbf{x}}}_2^2}$. 

Each process can access the same \emph{data distribution} $\dd$, 
and \emph{loss function} $\loss(\mathbf{x},z)$, 
which takes a \emph{learning parameter} $\mathbf{x}\in \R^d$ 
and a data point $z\in \dd$. 
Given a learning parameter $\mathbf{x}\in \R^d$, the \emph{cost function} $Q$ is:
\begin{equation*}
Q(\mathbf{x}) \triangleq \E_{z\sim \dd} \bracks*{\loss(\mathbf{x},z)}.
\end{equation*}
A distributed \emph{Stochastic Gradient Descent (SGD)} algorithm collectively 
minimizes the cost function $Q$, i.e., it finds 
$\mathbf{x}^*\in \argmin_{\mathbf{x}\in \R^d} Q(\mathbf{x})$.

The cost function is \emph{differentiable} and \emph{smooth}, 
i.e., for a constant $L\in \R^{+}$,
\begin{equation}
\label{l-lipsc}
        \forall \mathbf{x},\mathbf{y}\in \R^d,\, \norm*{\nabla Q(\mathbf{x})-\nabla Q(\mathbf{y})}_2\leq L\norm*{\mathbf{x} - \mathbf{y}}_2,
\end{equation}
where $\nabla Q(\mathbf{x})\in \R^d$ is the \emph{gradient} of $Q$ at $\mathbf{x}$.
The gradient at $\mathbf{x}\in \R^d$ can be estimated by 
the \emph{stochastic gradient} $G(\mathbf{x}, z) = \nabla \loss(\mathbf{x},z) \in \R^d$, 
calculated at a data point $z$ that is drawn uniformly at random from $\dd$. 
The stochastic gradient is an \emph{unbiased estimator} of the true gradient
\begin{equation}
\label{unbiased-est}
\E_{z\sim \dd} \bracks*{G(\mathbf{x},z)} = \nabla Q(\mathbf{x}).
\end{equation}
In addition, the estimations have \emph{bounded variance}, 
i.e., there is a non-negative constant $\sigma\in \R$ such that
\begin{equation}
\label{bounded-var}
\Var_{z\sim \dd} \bracks*{G(\mathbf{x}, z)} = \E_{z\sim \dd} \bracks*{\norm*{G(\mathbf{x}, z) - \nabla Q(\mathbf{x})}^2_2} \leq \sigma^2 .
\end{equation}
These are standard assumptions in SGD analysis~\cite{Bubeck15,bottou2018optimization,GhadimiL13a}.

At the end of the algorithm, each nonfaulty process $i$ outputs an 
estimate of the learning parameter, $\mathbf{x}^i \in \R^d$. 
We require the algorithm to \emph{externally converge 
with expected error $\epsconverge>0$} (simply called 
\emph{to converge} in the optimization literature). 
The convergence requirement expresses the quality of the solution 
relative to a minimal one. 
It varies according to the assumptions on the cost function, 
whether it is strongly-convex (Section~\ref{section:strongly convex}) 
or not (Section~\ref{section:non convex}).

The algorithm \emph{internally converge with expected error $\delta>0$},
if for every pair of nonfaulty processes $i$ and $j$ and valid execution tree of the algorithm $\mathcal{T}$:
\begin{equation}
    \label{diff-require}
    \E_{\mathcal{T}} \bracks*{\norm*{\mathbf{x}^i - \mathbf{x}^j}_2^2} \leq \epsmodel.
\end{equation}
We often omit the subscript of $\mathcal{T}$ from $\E$ when it is clear from the context.

As we only consider valid execution trees, we assume a {weak adversary} that does not observe the local state 
of the processes, and is oblivious to their local coin flips 
when scheduling the processes.
\section{Strongly-Convex Cost Functions}

\label{section:strongly convex}

We start by considering a strongly convex function $Q$. 
Formally, $Q$ is \emph{$\mu$-strongly convex}, for $\mu > 0$, 
if for every $\mathbf{x},\mathbf{y}\in \R^d$,
\begin{equation*}
    Q(\mathbf{y}) \geq Q(\mathbf{x}) + \tup*{\nabla Q(\mathbf{x}), \mathbf{y} - \mathbf{x}} + \frac{\mu}{2} \norm*{\mathbf{y} - \mathbf{x}}_2^2.
\end{equation*}
A strongly convex cost function has a single minimum, denoted $\mathbf{x}^*$.
In this case, an algorithm externally converges if
for every nonfaulty process $i$ and valid execution tree of the algorithm $\mathcal{T}$,
\begin{equation}
    \E_{\mathcal{T}} \bracks*{\norm*{\mathbf{x}^i - \mathbf{x}^*}_2^2} \leq \epsconverge.
\end{equation}
Following standard analysis~\cite{bottou2018optimization}, 
the convergence rate of vanilla SGD in this case
is $\mathcal{O}\parens*{1/{T}}$ after $T$ iterations.
That is, $\E \bracks*{\norm*{\mathbf{x}_T - \mathbf{x}^*}_2^2} \leq C/T$ 
for a constant $C$ depending on $\mu$, $L$, $\sigma$ 
and $\norm*{\mathbf{x}_1 - \mathbf{x}^*}_2^2$.
This implies that $\E \bracks*{\norm*{\mathbf{x}_T - \mathbf{x}^*}_2^2} \leq \epsconverge$
after $T=\mathcal{O}\parens*{\epsilon^{-1}}$ iterations.
Using a minibatch of size $b$ gives a linear speedup 
in convergence rate to $\mathcal{O}\parens*{{1}/b T}$.
This will serve as the baseline for the external convergence rate.

\begin{algorithm}[tb]
    \caption{
    \small{Cluster-based SGD, for strongly-convex function: code for process $i$}}
    \label{alg:dis-sgd-strongly-convex}
	\begin{algorithmic}[1]
	    \Statex \textbf{Local input:} initial point $\mathbf{x}^i_1$
		\For{$t=1\ldots T$}
		\State draw uniformly at random $z_t^i \in \dd$ 
        \label{lin:random}
		\State $\mathbf{g}_t^i \gets G\parens*{\mathbf{x}_t^i, z_t^i}$
		\State $\mathbf{y}_{t}^i \gets \mathbf{x}_t^i - \eta_t \mathbf{g}_t^i$
        \State broadcast $\tup*{t,\mathbf{y}_{t}^i}$ to all processes
        \label{lin:sent-msg}
        \State wait to receive $\thresh$ messages of the form $\tup{t,-}$
        \State $\mathbf{x}_{t+1}^i \gets \mathrm{avg}(\text{received learning parameters})$
        \label{lin:x_t+1}
        \EndFor
        \State output $\mathbf{x}_{T+1}^i$
	\end{algorithmic}
\end{algorithm}

Algorithm~\ref{alg:dis-sgd-strongly-convex} works in iterations,
corresponding to those of sequential SGD. 
A process starts an iteration $t$ with a local learning parameter,
and computes a new one for the next iteration.
First, the process computes a stochastic gradient using its current 
learning parameter,
performs a local SGD step,
and sends the updated learning parameter to all the other processes.
After receiving learning parameters for iteration $t$
from $\thresh$ processes, 
the process averages all the parameters and sets this value to be its learning parameter of the next iteration.
In the last iteration, each process outputs the last learning parameter
it has computed.
The algorithm ignores \emph{stale} gradients from previous iterations, 
i.e., gradients computed using learning parameter of a previous round.

The value of $\thresh$ can be arbitrary;
we only assume that $\thresh \leq n - \fp$, to ensure progress.
As we show, the larger $\thresh$ is, 
the better the convergence of the algorithm is.
The learning parameter for the first iteration of process $i$
is $\mathbf{x}_1^i\in \R^d$. Note that the learning parameters of the first iteration may be different across the processes.
The \emph{learning rate} for iteration $t$ is $\eta_t$, and for each $t$,
it is the same for all processes; 
the learning rate set in Algorithm~\ref{alg:dis-sgd-strongly-convex} is decreasing, i.e., 
$\eta_t = \mathcal{O} \parens*{1/t}$ for every iteration $t$.

Fix a valid execution tree $\mathcal{T}$, throughout the proof we simply use $\E$ instead of $\E_{\mathcal{T}}$.
Let $V_t$, $t \geq 1$, 
be the set of processes that compute learning parameters 
for iteration $t+1$ in Line~\ref{lin:x_t+1}.
$V_0$ is the set of processes that execute the first line.
These sets are well defined over $\mathcal{T}$, since the underlying process schedule of each execution in the tree is the same, and the randomness in Line~\ref{lin:random} only effects the values sent by the processes.
Note that $V_{t}\subseteq V_{t-1}$, for $t\geq 1$,  
and all processes that send a message in iteration $t$ 
(Line~\ref{lin:sent-msg}) are in $V_{t-1}$.
We use the notation $\E_t$ and $\Var_t$ for the expectation and variance over iteration $t$, 
i.e., the expectation only over the randomness of the stochastic gradients computed in iteration $t$. {Note that the statement and proof of some useful mathematical propositions are deferred to Appendix~\ref{app:use-props}.}

We can prove that the average of $b$ stochastic gradients, 
each computed using different learning parameters, 
has variance of at most ${\sigma^2}/{b}$. (See Appendix~\ref{app:add-proofs} for the proof.)

\begin{restatable}{lemma}{vartotal}
\label{var-total}
For any $\mathbf{x}_1, \ldots, \mathbf{x}_b \in \R^d$ and 
$z_1,\ldots, z_b$ drawn uniformly at random from $\dd$,
\begin{equation*}
    \Var \bracks*{\frac{1}{b}\sum_{i=1}^b G(\mathbf{x}_i,z_i)} \leq \frac{\sigma^2}{b}.
\end{equation*}
\end{restatable}
%
%
%
Since Line~\ref{lin:x_t+1} in Algorithm~\ref{alg:dis-sgd-strongly-convex}
averages the received parameters and the weak adversary cannot
determine the learning parameters that arrive at each process,
Lemma~\ref{var-total} implies:
\begin{lemma}
\label{var-total-strongly-convex}
For any iteration $t\geq 1$ {}and random set of processes $S\subseteq V_t$, 
$\Var \bracks*{\frac{1}{|S|} \sum_{i\in S} \mathbf{g}_t^i} \leq \frac{\sigma^2}{|S|}$
\end{lemma}

We first show that the diameter of the parameters at each iteration 
is contracted compared to the diameter of the previous iteration,
even if processes communicate with disjoint sets of processes, 
e.g., due to a system partition.
This contraction happens implicitly due to strong convexity
since intuitively, all the processes gravitate to the unique minimum.
Using this, we will show that the algorithm achieves internal convergence. For this purpose, we use the next lemma:
\begin{lemma}[{\cite[Lemma 6]{karimireddy2021scaffold}}]
\label{contractive-mapping}
    Let $Q$ be an $L$-smooth and $\mu$-strongly convex function and let $\eta \leq \frac{1}{L}$. Then for any $\mathbf{x},\mathbf{y}\in \R^d$,
    \begin{equation*}
        \norm*{\mathbf{x} - \eta \nabla Q(\mathbf{x}) - \mathbf{y} + \eta \nabla Q(\mathbf{y})}_2^2 \leq \parens*{1 - \eta \mu} \norm*{\mathbf{x} - \mathbf{y}}_2^2
    \end{equation*}
\end{lemma}

\begin{lemma}
\label{strong-convex-step-diam}
    Let $Q$ be an $L$-smooth and $\mu$-strongly convex function, 
    then for every iteration $t \geq 1$ where $\eta_t \leq \frac{1}{L}$,
    \begin{equation*}
        \max_{i,j\in V_t} \E \bracks*{\norm*{\mathbf{x}_{t+1}^i - \mathbf{x}_{t+1}^j}_2^2}
        \leq \parens*{1 - \eta_t\mu} \max_{i,j \in V_{t-1}} \E \bracks*{\norm*{\mathbf{x}_t^i - \mathbf{x}_t^j}_2^2} + \frac{4\sigma^2 \eta_t^2}{\thresh}
\end{equation*}
\end{lemma}

\begin{proof}
Consider two processes $i,j \in V_{t}$.
Let $S_1 = \{i_1,...,i_\thresh\}$ and $S_2 = \{j_1,...,j_\thresh\}$ be the sets of processes that were used to compute $\mathbf{x}^i_{t+1}$ and $\mathbf{x}^j_{t+1}$ in Line~\ref{lin:x_t+1}, respectively. 
\begin{align*}
    & \E_t \bracks*{\norm*{\mathbf{x}_{t+1}^{i} - \mathbf{x}_{t+1}^{j}}_2^2}
    = \E_t \bracks*{\norm*{\frac{1}{\thresh} \sum_{k\in S_1} \mathbf{y}_{t}^k - \frac{1}{\thresh} \sum_{k\in S_2} \mathbf{y}_{t}^k}_2^2} \\
    &= \E_t \bracks*{\norm*{\frac{1}{\thresh} \sum_{k=1}^\thresh \mathbf{x}_t^{i_k} - \frac{\eta_t}{\thresh} \sum_{k=1}^\thresh \mathbf{g}_t^{i_k} - \frac{1}{\thresh} \sum_{k=1}^\thresh \mathbf{x}_t^{j_k} + \frac{\eta_t}{\thresh} \sum_{k=1}^\thresh \mathbf{g}_t^{j_k}}_2^2} \\
    &= \norm*{\frac{1}{\thresh} \sum_{k=1}^\thresh \parens*{ \mathbf{x}_t^{i_k} - \eta_t \E_t \bracks*{\mathbf{g}_t^{i_k}} - \mathbf{x}_t^{j_k} + {\eta_t} \E_t \bracks*{\mathbf{g}_t^{j_k}}}}_2^2 + \Var \bracks*{\frac{\eta_t}{\thresh} \sum_{k=1}^\thresh \mathbf{g}_t^{j_k} - \frac{\eta_t}{\thresh} \sum_{k=1}^\thresh \mathbf{g}_t^{i_k}} && \explain{by  Proposition~\ref{prop:sep-mean-var}} \\
    &\leq \norm*{\frac{1}{\thresh} \sum_{k=1}^\thresh \parens*{ \mathbf{x}_t^{i_k} - \eta_t \E_t \bracks*{\mathbf{g}_t^{i_k}} - \mathbf{x}_t^{j_k} + {\eta_t} \E_t \bracks*{\mathbf{g}_t^{j_k}}}}_2^2 + 2\Var \bracks*{\frac{\eta_t}{\thresh} \sum_{k=1}^\thresh \mathbf{g}_t^{j_k}} + 2 \Var \bracks*{\frac{\eta_t}{\thresh} \sum_{k=1}^\thresh \mathbf{g}_t^{i_k}} && \explain{by Proposition~\ref{prop:var-diff}} \\
    &\leq \norm*{\frac{1}{\thresh} \sum_{k=1}^\thresh \parens*{\mathbf{x}_t^{i_k} - \eta_t \nabla Q(\mathbf{x}_t^{i_k}) - \mathbf{x}_t^{j_k} + \eta_t \nabla Q(\mathbf{x}_t^{j_k})}}_2^2 + \frac{4\sigma^2\eta_t^2}{\thresh} \\
    &\leq \frac{1}{\thresh} \sum_{k=1}^\thresh  \norm*{\mathbf{x}_t^{i_k} - \eta_t \nabla Q(\mathbf{x}_t^{i_k}) - \mathbf{x}_t^{j_k} + \eta_t \nabla Q(\mathbf{x}_t^{j_k})}_2^2 + \frac{4\sigma^2\eta_t^2}{\thresh} && \explain{by \eqref{relaxed-triangle-ineq2}} 
\end{align*}
In the second to last inequality we use \eqref{unbiased-est} for the expectation value and Lemma~\ref{var-total-strongly-convex} for the variance value.
Therefore,
\begin{align*}
    &\E \bracks*{\norm*{\mathbf{x}_{t+1}^{i} - \mathbf{x}_{t+1}^{j}}_2^2}
    = \E \bracks*{\E_t \bracks*{\norm*{\mathbf{x}_{t+1}^{i} - \mathbf{x}_{t+1}^{j}}_2^2}} \\
    &\leq \frac{1}{\thresh} \sum_{k=1}^\thresh  \E \bracks*{\norm*{\mathbf{x}_t^{i_k} - \eta_t \nabla Q(\mathbf{x}_t^{i_k}) - \mathbf{x}_t^{j_k} + \eta_t \nabla Q(\mathbf{x}_t^{j_k})}_2^2} + \frac{4\sigma^2\eta_t^2}{\thresh}\\
    &\leq \frac{1}{\thresh} \sum_{k=1}^\thresh \parens*{1 - \eta_t \mu} \E \bracks*{\norm*{\mathbf{x}_t^{i_k} - \mathbf{x}_t^{j_k}}_2^2} + \frac{4\sigma^2 \eta_t^2}{\thresh} && \explain{by Lemma~\ref{contractive-mapping}}\\
    &\leq \parens*{1 - \eta_t \mu} \max_{k,l\in V_t} \E \bracks*{\norm*{\mathbf{x}_t^{k} - \mathbf{x}_t^{l}}_2^2} + \frac{4\sigma^2 \eta_t^2}{\thresh}
\end{align*}
\end{proof}

The next lemma is useful to prove convergence for strongly convex functions. 
\begin{lemma}[{\cite[Lemma 9]{Alistarh2021-ElasticConsistency}}]
\label{eq:strong-convex}
    Let $Q$ be an $L$-smooth and $\mu$-strongly convex function with a single minimum $\mathbf{x}^*$, then for any $\mathbf{x} \in \R^d$
    \begin{equation*}
        \tup{\nabla Q(\mathbf{x}), \mathbf{x} - \mathbf{x}^*} \geq \frac{1}{2L}\norm*{\nabla Q(\mathbf{x})}_2^2 + \frac{\mu}{2} \norm*{\mathbf{x} - \mathbf{x}^*}_2^2
    \end{equation*}
\end{lemma}

The external convergence rate in the strongly-convex case does not depend on the diameter of the learning parameters at the same iteration:
\begin{lemma}
\label{strong-convex-step}
    Let $Q$ be an $L$-smooth and $\mu$-strongly convex function with a single minimum $\mathbf{x}^*$. Then, for every iteration $t\geq 1$ where $\eta_t \leq \frac{1}{L}$,
    \begin{equation*}
        \max_{i\in V_t} \E \bracks*{\norm*{\mathbf{x}^i_{t+1} - \mathbf{x}^*}_2^2} \leq \parens*{1 - \eta_t \mu} \max_{i\in V_{t-1}} \E \bracks*{\norm*{\mathbf{x}_t^i - \mathbf{x}^*}_2^2} + \frac{\sigma^2 \eta_t^2}{\thresh}
    \end{equation*}
\end{lemma}

\begin{proof}
Let process $i\in V_t$ and let $S$ be the set of $\thresh$ processes 
that was used to compute $\mathbf{x}^i_{t+1}$ in Line~\ref{lin:x_t+1}.
    \begin{align*}
        &\E_t \bracks*{\norm*{\mathbf{x}^i_{t+1} - \mathbf{x}^*}_2^2} 
        = \E_t \bracks*{\norm*{\frac{1}{\thresh}\sum_{j\in S} \mathbf{y}_{t}^j - \mathbf{x}^*}_2^2}\\
        &= \E_t \bracks*{\norm*{\frac{1}{\thresh}\sum_{j\in S} \mathbf{x}_t^{j} - \frac{\eta_t}{\thresh}\sum_{j\in S} \mathbf{g}_t^j - \mathbf{x}^*}_2^2}\\
         &= \norm*{\frac{1}{\thresh}\sum_{j\in S} \mathbf{x}_t^{j} - \E_t \bracks*{\frac{\eta_t}{\thresh}\sum_{j\in S} \mathbf{g}_t^j} - \mathbf{x}^*}_2^2 + \Var \bracks*{\frac{\eta_t}{\thresh}\sum_{j\in S} \mathbf{g}_t^j} && \explain{by Proposition~\ref{prop:sep-mean-var}}\\
        &\leq \norm*{\frac{1}{\thresh}\sum_{j\in S} \mathbf{x}_t^{j} - \frac{\eta_t}{\thresh}\sum_{j\in S} \nabla Q(\mathbf{x}_t^{j}) - \mathbf{x}^*}_2^2 + \frac{\eta_t^2 \sigma^2}{\thresh} && \explain{by Lemma~\ref{var-total-strongly-convex} and \eqref{unbiased-est}}\\ 
        &\leq \frac{1}{\thresh}\sum_{j\in S} \norm*{\mathbf{x}_t^{j} - \eta_t \nabla Q(\mathbf{x}_t^{j}) - \mathbf{x}^*}_2^2 + \frac{\eta_t^2 \sigma^2}{\thresh} && \explain{by \eqref{relaxed-triangle-ineq2}}
    \end{align*}
    Note that
    \begin{align*}
        &\norm*{\mathbf{x}_t^{j} - \eta_t \nabla Q(\mathbf{x}_t^{j}) - \mathbf{x}^*}_2^2 \\
        &= \norm*{\mathbf{x}_t^j - \mathbf{x}^*}_2^2 - 2\eta_t \tup{\nabla Q(\mathbf{x}_t^j),\mathbf{x}_t^j - \mathbf{x}^*} + \eta_t^2 \norm*{\nabla Q(\mathbf{x}_t^j)}_2^2 \\
        &\leq \parens*{1 - \eta_t \mu}\norm*{\mathbf{x}_t^j - \mathbf{x}^*}_2^2 + \parens*{\eta_t^2 - \frac{\eta_t}{L}} \norm*{\nabla Q(\mathbf{x}_t^j)}_2^2 && \explain{by Lemma~\ref{eq:strong-convex}} \\
        &\leq \parens*{1 - \eta_t \mu}\norm*{\mathbf{x}_t^j - \mathbf{x}^*}_2^2 && \explain{$\eta_t \leq \frac{1}{L}$}
    \end{align*}
    Therefore,
    \begin{align*}
        \E \bracks*{\norm*{\mathbf{x}^i_{t+1} - \mathbf{x}^*}_2^2} 
        &= \E \bracks*{\E_t \bracks*{\norm*{\mathbf{x}^i_{t+1} - \mathbf{x}^*}_2^2}}\\
        &\leq \frac{1}{\thresh}\sum_{j\in S} \parens*{1 - \eta_t \mu} \E \bracks*{\norm*{\mathbf{x}_t^j - \mathbf{x}^*}_2^2} + \frac{\eta_t^2 \sigma^2}{\thresh} \\
        &\leq \parens*{1 - \eta_t \mu} \max_{j\in V_{t-1}} \E \bracks*{\norm*{\mathbf{x}_t^j - \mathbf{x}^*}_2^2} + \frac{\eta_t^2 \sigma^2}{\thresh}
    \end{align*}
\end{proof}

Note that the terms we get in Lemma~\ref{strong-convex-step-diam} 
and Lemma~\ref{strong-convex-step} are very similar to the ones obtained 
in~\cite[Section 4.2]{bottou2018optimization} for the strongly-convex case 
using mini-batches of size $\thresh$.
However, we bound the difference between the learning parameters themselves,
while they bound the difference between the function values 
at the learning parameters. 
The next lemma is adapted from \cite[Theorem~4.7]{bottou2018optimization} and the full proof appears in Appendix~\ref{app:add-proofs}.
\begin{restatable}{lemma}{bottouoneovert}
\label{bottou-1/t}
If for any $t\geq 1$,
\begin{equation*}
    a_{t+1} \leq \parens*{1-\eta_t \mu} a_t + \frac{c\eta_t^2}{b}
\end{equation*}
for some $c,\mu > 0$ and $b\geq1$. Then, for decreasing learning rate $\eta_t$, such that for all $t\geq 1$,
    $\eta_t = \frac{\beta}{\gamma + t}$
for some $\beta > \frac{1}{\mu}$ and $\gamma > 0$,
\begin{equation*}
    a_t \leq \frac{\nu}{b\parens*{\gamma + t}},
\end{equation*}
where $\nu = \max\braces*{\frac{\beta^2 c}{\beta\mu - 1},\parens*{\gamma + 1}a_1}$.
\end{restatable}

The next theorem shows that for strongly-convex functions our algorithm 
achieves the same convergence rate (both internal and external) 
as the sequential baseline. 
Intuitively, as there is a single minimum, all processes will converge 
to this point independently. 
Hence, despite different processes holding different learning parameters, 
the expected gradients will point in the same direction. 
This allows us to obtain terms similar to the ones in the classical 
strongly-convex analysis~\cite{bottou2018optimization}, 
and achieve the same convergence rate. 
Lemma~\ref{strong-convex-step-diam} and Lemma~\ref{strong-convex-step},
together with Lemma~\ref{bottou-1/t}, give our first main result: 

\begin{theorem}
\label{thm:strongly-convex}
Let $Q$ be an $L$-smooth and $\mu$-strongly convex function with a single minimum $\mathbf{x}^*$, then for decreasing learning rate $\eta_t = \frac{\beta}{\gamma + t} \leq \frac{1}{L}$ for some constants $\beta > \frac{1}{\mu}$ and $\gamma > 0$,
\begin{align*}
    \max_{i,j\in V_T} \E \bracks*{\norm*{\mathbf{x}^{i}_{T+1} - \mathbf{x}^{j}_{T+1}}_2^2}
    &\leq  \frac{\parens*{\gamma + 1} \max_{i,j\in V_0} \norm*{\mathbf{x}_1^i - \mathbf{x}_1^j}_2^2}{\parens*{\gamma + T + 1}\thresh} + \frac{4\beta^2 \sigma^2}{\parens*{\beta \mu - 1} \parens*{\gamma + T + 1}\thresh{}} 
    \tag{Internal convergence}\\
    \max_{i\in V_T} \E \bracks*{\norm*{\mathbf{x}^i_{T+1} - \mathbf{x}^*}_2^2} 
    &\leq  \frac{\parens*{\gamma + 1} \max_{i\in V_0} \norm*{\mathbf{x}_1^i - \mathbf{x}^*}_2^2}{\parens*{\gamma + T + 1}\thresh} + \frac{\beta^2 \sigma^2}{\parens*{\beta \mu - 1} \parens*{\gamma + T + 1}\thresh{}}
    \tag{External convergence}
\end{align*}
\end{theorem}

For the special case where the processes start at the same initial point $\mathbf{x}_1$ we have:
\begin{corollary}[Algorithm~\ref{alg:dis-sgd-strongly-convex} with a single initial point]
    Let $Q$ be an $L$-smooth and $\mu$-strongly convex function with a single minimum $\mathbf{x}^*$, then for initial point $\mathbf{x}_1$ such that $\mathbf{x}_1^i = \mathbf{x}_1$ for every $1\leq i \leq n$ and decreasing learning rate $\eta_t = \frac{\beta}{\gamma + t} \leq \frac{1}{L}$ for some constants $\beta > \frac{1}{\mu}$ and $\gamma > 0$,
\begin{align*}
    \max_{i,j\in V_T} \E \bracks*{\norm*{\mathbf{x}^{i}_{T+1} - \mathbf{x}^{j}_{T+1}}_2^2}
    &\leq \frac{4\beta^2 \sigma^2}{\parens*{\beta \mu - 1} \parens*{\gamma + T + 1}\thresh{}} 
    \tag{Internal convergence}\\
    \max_{i\in V_T} \E \bracks*{\norm*{\mathbf{x}^i_{T+1} - \mathbf{x}^*}_2^2} 
    &\leq  \frac{\parens*{\gamma + 1} \norm*{\mathbf{x}_1 - \mathbf{x}^*}_2^2}{\parens*{\gamma + T + 1}\thresh} + \frac{\beta^2 \sigma^2}{\parens*{\beta \mu - 1} \parens*{\gamma + T + 1}\thresh{}}
    \tag{External convergence}
\end{align*}
\end{corollary}
Neglecting dependencies on 
$\max_{i,j\in V_0} \norm*{\mathbf{x}_1^i - \mathbf{x}_1^j}_2^2$, $\max_{i\in V_0} \norm*{\mathbf{x}_1^i - \mathbf{x}^*}_2^2$,
$\mu$, $L$ and $\sigma$, this means that Algorithm~\ref{alg:dis-sgd-strongly-convex} converges externally 
in $\mathcal{O} \parens*{\epsconverge^{-1}/\thresh}$ iterations, 
and internally in $\mathcal{O} \parens*{\epsmodel^{-1}/\thresh}$ iterations;
that is, the rates are the same.
Since each iteration takes a single communication round, 
we get the same upper bound on the number of rounds.

\section{Non-Convex Cost Functions}
\label{section:non convex}

In the general case, where the function $Q$ is non-convex, 
the algorithm has to converge to a point with zero gradient; 
that is, for every nonfaulty process $i$ and valid execution tree $\mathcal{T}$ of the algorithm:
\begin{equation}
    \label{convergence-require}
    \E_{\mathcal{T}} \bracks*{\norm*{\nabla Q (\mathbf{x}^i)}_2^2} \leq \epsconverge.
\end{equation}
We assume that $Q$ is lower bounded by $Q^*$, 
i.e., for every $\mathbf{x}\in \R^d$, $Q(\mathbf{x}) \geq Q^*$~\cite{bottou2018optimization}.

The convergence rate of (vanilla) minibatch SGD with batch 
size $b$ is $\mathcal{O}\parens*{1/\sqrt{b T}}$, 
i.e., $\E \bracks*{\norm*{\nabla Q(\mathbf{x}_\tau)}_2^2} \leq C/\sqrt{b T}$, for iteration $\tau$ drawn uniformly at random from $[T]$ and constant $C$ which depends on $L$, $\sigma$ and $\parens*{Q(\mathbf{x}_1) - Q^*}$~\cite{GhadimiL13a}. This implies that $\E \bracks*{\norm*{\nabla Q(\mathbf{x}_\tau)}_2^2} \leq \epsconverge$ after $T = \mathcal{O}\parens{\epsconverge^{-2}/\sqrt{b}}$ iterations.
As in the strongly-convex case, this will serve as our baseline.

When the function $Q$ is not strongly convex, 
processes that obtain disjoint estimations at an iteration may compute \emph{diverging} learning parameters.
For this reason, we need to ensure that processes communicate
with intersecting sets of \emph{clusters}. 
To further expedite the contraction rate, 
and reduce the distance between the learning parameters, 
we end each iteration with \emph{multidimensional 
approximate agreement} (\emph{MDAA}). 
The input to MDAA is the local learning parameter, 
and its output serves as the learning parameter
for the next iteration.

Formally, in \emph{multidimensional approximate agreement}~\cite{MendesH13},
each process $i$ starts with input $\mathbf{x}_i \in \R^d$ 
and outputs a value $\mathbf{y}_i \in \R^d$, such that:
\begin{description}
\item[Convexity:] 
The outputs are in the \emph{convex hull} of the inputs,
that is, they are a \emph{convex combination} of the outputs.
\item[$q$-Contraction:]
The outputs are contracted by a factor of $q$ relative to the inputs,
that is, for every pair of nonfaulty processes $i,j$,
$\norm*{\mathbf{y}_i - \mathbf{y}_j}_2^2 
\leq q \, \diameter{\braces*{\mathbf{x_1},\dots\mathbf{x}_n}}$,
where the squared Euclidean \emph{diameter} of a set $A\subseteq \R^d$ is 
$\diameter{A} \triangleq \max_{\mathbf{x},\mathbf{y}\in A} \norm*{\mathbf{x} - \mathbf{y}}_2^2$.
\end{description}
Standard approximate agreement~\cite{MendesH13,DolevLPSW86}  
requires \emph{$\epsilon$-agreement}, that is, 
the distance between outputs is at most $\epsilon$. 
We only require contraction \emph{relative to the diameter of the inputs},
rather than a predefined maximal distance.
Section~\ref{sec:approximate agreement} presents a cluster-based MDAA algorithm assuming $f \leq \fopt$, 
which requires $\mathcal{O}\parens*{\log q^{-1}}$ communication rounds to achieve $q$ contraction.

\begin{algorithm}[tb]
   \caption{
    Cluster-based SGD: code for process $i$}
    \label{alg:dis-sgd}
	\begin{algorithmic}[1]
        \Statex \textbf{Global input:} initial point $\mathbf{x}_1$ and random iteration $\tau$
	    \State $\mathbf{x}^i_1 \gets \mathbf{x}_1$
		\For{$t=1\dots T$}
		\State draw uniformly at random $z_t^i \in \dd$ 
		 \State broadcast $\tup*{t,G\parens*{\mathbf{x}_t^i, z_t^i}}$ to all processes
		 \State wait to receive $\thresh$ messages of the form $\tup{t,-}$
		\State $\mathbf{g}_t^i \gets \mathrm{avg}(\text{received stochastic gradients})$
		\label{lin:avg-grad}
        \State $\mathbf{y}_t^i \gets \mathbf{x}_t^i - \eta_t \mathbf{g}_t^i$
        \State $\mathbf{x}_{t+1}^i \gets $ {\sf MDAA}$_t(\mathbf{y}_t^i,  q)$
        \label{lin:approx-agree}
        \EndFor
        \State output $\mathbf{x}_\tau^i$ 
	\end{algorithmic}
\end{algorithm}
Algorithm~\ref{alg:dis-sgd} 
deals with non-convex functions. 
One difference from Algorithm~\ref{alg:dis-sgd-strongly-convex} 
is in Line~\ref{lin:approx-agree}, 
calling MDAA with contraction parameter $q$.
Another difference is that processes send the stochastic gradients they computed, 
average the received gradients to a \emph{mini-batch} stochastic gradient 
and then use it to perform a local SGD step.
Every process outputs the learning parameter of the \emph{same} iteration $\tau\in [T]$, 
which is drawn uniformly at random;
this is a typical practice in SGD algorithms for 
non-convex objective functions~\cite{bottou2018optimization,GhadimiL13a}. The convergence proof of the non-convex algorithm uses a constant learning rate, i.e., $\eta_t = \eta$ for all $t$ and constant $\eta$.

Algorithm~\ref{alg:dis-sgd} considers the case in which all processes start with the same initial point $\mathbf{x}_1 \in \R^d$. 
Later, in Appendix~\ref{sec:differentinitialpoints}, we discuss the case where the processes start at different initial points.

The algorithm is similar to those in~\cite{CollaborativeLearningNeurIPS},
but they use aggregation rules which are resilient against malicious failures, and \emph{averaging agreement},
while we rely on the convexity property ensured by MDAA.

The parameter $\thresh{} \in [n]$ determines how many messages a process 
waits for in every iteration;
to ensure progress, we require that $\thresh{} \leq  n - \fp$, 
as in the strongly-convex case.
Furthermore, we assume that $\fp \leq \fopt$ to guarantee the convergence of the MDAA algorithm.

Fix a valid execution tree $\mathcal{T}$, as in Section~\ref{section:strongly convex}, we omit the subscript of $\mathcal{T}$ throughout the proof.
Let $V_t$, $t \geq 1$, be the set of processes that compute learning parameters 
for iteration $t+1$ (Line~\ref{lin:approx-agree}) and $V_0$ be the set of processes that execute the first line.
Similarly to the strongly-convex case, following Lemma~\ref{var-total}, we have the next lemma:
\begin{lemma}
    \label{var-total-non-convex}
For every iteration $t\geq 1$ and process $i\in V_t$,
    $\Var \bracks*{\mathbf{g}_t^i} \leq \frac{\sigma^2}{\thresh{}}$.
\end{lemma}

The proof for external convergence uses the
\emph{effective gradient}~\cite{El-MhamdiGGHR20}, 
which is defined as the difference between two consecutive iterations parameters.
Formally, the effective gradient of iteration $t$ and process $i\in V_t$ is:
\begin{equation}
\label{eq:effective-grad}
    \mathbf{G}_t^i \triangleq \frac{\mathbf{x}_t^i - \mathbf{x}_{t+1}^i}{\eta_t} .
\end{equation}
We bound the difference between the effective gradient and 
the true gradient in each iteration,
depending on the diameter of the learning parameters in the same iteration. 
This allows to prove convergence, 
as the effective change between two consecutive iterations is 
close enough to the true gradient.
The next lemma 
is the analog of Lemma~\ref{var-total-non-convex} for the effective gradient.

\begin{lemma}
\label{effect-grad-var}
For every iteration $t\geq 1$ and process $i\in V_t$,
\begin{equation*}
    \Var_t \bracks*{\mathbf{G}_t^i} \leq \frac{\sigma^2}{\thresh}.
\end{equation*}
\end{lemma}

\begin{proof}
By the convexity property of the approximate agreement algorithm, 
$\mathbf{x}^i_{t+1} = \sum_{j=1}^k w_j \mathbf{y}_t^{i_j}$ 
for some $\braces*{i_1,\dots,i_k}\subseteq V_t$ and weights $\sum_{j=1}^k w_j = 1$.
\begin{equation*}
    \mathbf{x}^i_t - \mathbf{x}^i_{t+1}
    = \mathbf{x}^i_t - \sum_{j=1}^k w_j \mathbf{y}_t^{i_j} = \mathbf{x}^i_t - \sum_{j=1}^k w_j \parens*{\mathbf{x}_t^{i_j} - \eta_t \mathbf{g}_t^{i_j}}
\end{equation*}
Hence,
    \begin{align*}
        \Var_t \bracks*{\mathbf{G}_t^i} 
        &= \E_t \bracks*{\norm*{\mathbf{G}_t^i - \E_t \bracks*{\mathbf{G}^i_t}}_2^2} \\
        &= \E_t \bracks*{\norm*{\frac{1}{\eta_t} \parens*{\mathbf{x}_t^i - \mathbf{x}_{t+1}^i} - \E_t \bracks*{\frac{1}{\eta_t} \parens*{\mathbf{x}_t^i - \mathbf{x}_{t+1}^i}}}_2^2} \\
        &= \E_t \bracks*{\norm*{\sum_{j=1}^k w_j \mathbf{g}_t^{i_j} - \sum_{j=1}^k w_j \E_t \bracks*{\mathbf{g}_t^{i_j}}}_2^2}\\
        &\leq \sum_{j=1}^k w_j \E_t \bracks*{\norm*{ \mathbf{g}_t^{i_j} - \E_t \bracks*{\mathbf{g}_t^{i_j}}}_2^2} && \explain{by Proposition~\ref{prop:sep-weights}} \\
        &= \sum_{j=1}^k w_j  \Var \bracks*{\mathbf{g}_t^{i_j}} \leq \frac{\sigma^2}{\thresh} && \explain{by Lemma~\ref{var-total-non-convex}}
    \end{align*}
\end{proof}

A key part of our analysis is separating the mean and the variance of the 
effective gradient when plugging it in the standard analysis of SGD. 
The convexity property of MDAA allows to prove the next
lemma, showing that the expectation of the difference between 
the expectation of the effective gradient 
and the true gradient does not depend on $\sigma$.

\begin{lemma}
\label{effect-grad-diff}
For every iteration $t\geq 1$ and process $i\in V_t$,
\begin{equation*}
    \E \bracks*{\norm*{\E_t \bracks*{\mathbf{G}^i_t} - \nabla Q(\mathbf{x}_t^i)}_2^2} \leq \parens*{\frac{2}{\eta_t^2} + 2L^2}\max_{i,j\in V_t} \E \bracks*{\norm*{\mathbf{x}_t^i - \mathbf{x}^j_t}_2^2}
\end{equation*}
\end{lemma}

\begin{proof}
By convexity of MDAA, 
$\mathbf{x}^i_{t+1} = \sum_{j=1}^k w_j \mathbf{y}_t^{i_j}$ for 
some $\braces*{i_1,\dots,i_k}\subseteq V_t$ and weights $\sum_{j=1}^k w_j = 1$.
\begin{align*}
        &\norm*{\E_t \bracks*{\mathbf{G}^i_t} - \nabla Q(\mathbf{x}_t^i)}_2^2 \\
        &= \norm*{\frac{1}{\eta_t}\parens*{\mathbf{x}^i_t - \sum_{j=1}^k w_j \mathbf{x}_t^{i_j}} + \sum_{j=1}^k w_j \E_t \bracks*{\mathbf{g}_t^{i_j}} - \nabla Q(\mathbf{x}_t^i)}_2^2 \\
        &\leq \frac{2}{\eta_t^2} \norm*{\mathbf{x}^i_t - \sum_{j=1}^k w_j \mathbf{x}_t^{i_j}}_2^2 + 2\norm*{\sum_{j=1}^k w_j \nabla Q(\mathbf{x}_t^{i_j}) - \nabla Q(\mathbf{x}_t^i)}_2^2 && \explain{by \eqref{unbiased-est} and \eqref{relaxed-triangle-ineq2}}\\
        &\leq \frac{2}{\eta_t^2} \sum_{j=1}^k w_j \norm*{\mathbf{x}^i_t - \mathbf{x}_t^{i_j}}_2^2 + 2 \sum_{j=1}^k w_j \norm*{\nabla Q(\mathbf{x}_t^{i_j}) - \nabla Q(\mathbf{x}_t^i)}_2^2 &&  \explain{by Proposition~\ref{prop:sep-weights}} \\
        &\leq \frac{2}{\eta_t^2} \sum_{j=1}^k w_j \norm*{\mathbf{x}^i_t - \mathbf{x}_t^{i_j}}_2^2 + 2 L^2 \sum_{j=1}^k w_j \norm*{\mathbf{x}_t^{i_j} - \mathbf{x}_t^i}_2^2 &&  \explain{by \eqref{l-lipsc}}
\end{align*}
Therefore,
\begin{align*}
    \E \bracks*{\norm*{\E_t \bracks*{\mathbf{G}^i_t} - \nabla Q(\mathbf{x}_t^i)}_2^2} 
    &\leq \parens*{\frac{2}{\eta_t^2} + 2 L^2} \sum_{j=1}^k w_j \E \bracks*{\norm*{\mathbf{x}^i_t - \mathbf{x}_t^{i_j}}_2^2} \\
    &\leq \parens*{\frac{2}{\eta_t^2} + 2L^2}\max_{i,j\in V_t} \E \bracks*{\norm*{\mathbf{x}_t^i - \mathbf{x}^j_t}_2^2} && \explain{$\sum_{j=1}^k w_j = 1$}
\end{align*}
\end{proof}

The next lemma uses the following upper bound on $Q$, which is implied from the smoothness of the function~\eqref{l-lipsc}:
\begin{equation}
\label{descent-lemma}
        \forall \mathbf{x},\mathbf{y}\in \R^d,\, Q(\mathbf{y}) \leq Q(\mathbf{x}) + \tup{\nabla Q(\mathbf{x}), \mathbf{y} - \mathbf{x}} + \frac{L}{2} \norm*{\mathbf{y} - \mathbf{x}}_2^2.
\end{equation}

The next lemma is similar to the one used to prove the convergence 
of non-convex vanilla SGD~\cite{bottou2018optimization}, 
when the effective gradient is replaced with a stochastic gradient 
computed with respect to $\mathbf{x}_t^i$.
For vanilla SGD, where gradients are unbiased, the term 
$\norm*{\E_t\bracks*{\mathbf{G}_t^i} - \nabla Q(\mathbf{x}_t^i)}_2^2$ becomes 0. 
In our case, by Lemma~\ref{effect-grad-diff}, this term depends on 
the distance between the learning parameters, which we bound using MDAA.

\begin{lemma}
\label{non-convex-step}
    Let $Q$ be an $L$-smooth function. For every iteration $t\geq1$ where $\eta_t\leq \frac{1}{4L}$ and process $i\in V_t$,
    \begin{equation*}
        \frac{\eta_t}{4} \norm*{\nabla Q(\mathbf{x}_t^i)}_2^2 \leq Q(\mathbf{x}_{t}^i) - \E_t \bracks*{Q(\mathbf{x}_{t+1}^i)} + \eta_t \norm*{\E_t\bracks*{\mathbf{G}_t^i} - \nabla Q(\mathbf{x}_t^i)}_2^2 + \frac{\sigma^2\eta_t^2 L}{\thresh}
    \end{equation*}
\end{lemma}

\begin{proof}
By \eqref{eq:effective-grad}, $\mathbf{x}_{t+1}^i = \mathbf{x}_t^i - \eta_t \mathbf{G}_t^i$.
By using \eqref{descent-lemma}
\begin{equation}
\label{eq:51}
    \E_t\bracks*{Q(\mathbf{x}_{t+1}^i)}
    \leq Q(\mathbf{x}_{t}^i) -\eta_t \E_t\bracks*{\tup{\mathbf{G}_t^i,\nabla Q(\mathbf{x}_t^i)}} + \frac{\eta_t^2 L}{2} \E_t\bracks*{\norm*{\mathbf{G}_t^i}_2^2}
\end{equation}
For the second term
\begin{equation}
\label{eq:51-1}
\begin{aligned}
     -\E_t\bracks*{\tup{\mathbf{G}_t^i,\nabla Q(\mathbf{x}_t^i)}}
     &= -\tup{\E_t\bracks*{\mathbf{G}_t^i},\nabla Q(\mathbf{x}_t^i)} \\
     &= -\norm*{\nabla Q(\mathbf{x}_t^i)}_2^2 + \tup{\nabla Q(\mathbf{x}_t^i) - \E_t\bracks*{\mathbf{G}_t^i},\nabla Q(\mathbf{x}_t^i)} && \explain{by \eqref{inned-prod-add}}\\
     &\leq - \norm*{\nabla Q(\mathbf{x}_t^i)}_2^2 + \frac{1}{2} \norm*{\nabla Q(\mathbf{x}_t^i) - \E_t\bracks*{\mathbf{G}_t^i}}_2^2 + \frac{1}{2} \norm*{\nabla Q(\mathbf{x}_t^i)}_2^2 && \explain{by \eqref{eq:young}}
\end{aligned}
\end{equation}
For the third term
\begin{equation}
    \label{eq:51-2}
\begin{aligned}
    &\frac{1}{2} \E_t\bracks*{\norm*{\mathbf{G}_t^i}_2^2} \\
    &= \frac{1}{2} \E_t\bracks*{\norm*{\mathbf{G}_t^i - \nabla Q(\mathbf{x}_t^i) + \nabla Q(\mathbf{x}_t^i)}_2^2} \\
    &\leq \E_t\bracks*{\norm*{\mathbf{G}_t^i - \nabla Q(\mathbf{x}_t^i)}_2^2} + \norm*{\nabla Q(\mathbf{x}_t^i)}_2^2 && \explain{by \eqref{relaxed-triangle-ineq2}} \\
    &\leq \norm*{\E_t\bracks*{\mathbf{G}_t^i} - \nabla Q(\mathbf{x}_t^i)}_2^2 + \norm*{\nabla Q(\mathbf{x}_t^i)}_2^2 + \frac{\sigma^2}{\thresh} && \explain{by Proposition~\ref{prop:sep-mean-var} and Lemma~\ref{effect-grad-var}}
\end{aligned}
\end{equation}
Combining \eqref{eq:51}, \eqref{eq:51-1} and \eqref{eq:51-2}, and since $\eta_t \leq \frac{1}{4L}$
\begin{multline*}
    \E_t\bracks*{Q(\mathbf{x}_{t+1}^i)} \\
    \leq Q(\mathbf{x}_{t}^i) + \parens*{\eta_t^2 L - \frac{\eta_t}{2}} \norm*{\nabla Q(\mathbf{x}_t^i)}_2^2 + \parens*{\frac{\eta_t}{2} + \eta_t^2 L}  \norm*{\E_t\bracks*{\mathbf{G}_t^i} - \nabla Q(\mathbf{x}_t^i)}_2^2 + \frac{\sigma^2\eta_t^2 L}{\thresh} \\
    \leq Q(\mathbf{x}_{t}^i) - \frac{\eta_t}{4} \norm*{\nabla Q(\mathbf{x}_t^i)}_2^2 + \eta_t \norm*{\E_t \bracks*{\mathbf{G}_t^i} - \nabla Q(\mathbf{x}_t^i)}_2^2 + \frac{\sigma^2\eta_t^2 L}{\thresh}
\end{multline*}
The lemma follows by rearrangement. 
\end{proof}

The next lemma shows that external convergence depends on the maximal distance 
between the learning parameters at the different processes, 
and uses a fixed learning rate of $\eta = \sqrt{\thresh/T}$.
Since $\tau$ is drawn uniformly at random, 
the expectation of $\norm*{\nabla Q(\mathbf{x}_\tau^i)}_2^2$ is equal to the average expectation over all iterations.
Note that the first two terms are the classical error rates in the non-convex case~\cite{bottou2018optimization}.

\begin{lemma}
\label{thm:non-convex}
    Let $Q$ be an $L$-smooth cost function. Then for $T\geq 16L^2\thresh{}$, 
    constant learning rate $\eta = \frac{\sqrt{\thresh{}}}{\sqrt{T}}$ 
    and any process $i\in V_T$,
    \begin{equation*}
        \E \bracks*{\norm*{\nabla Q(\mathbf{x}_\tau^i)}_2^2} 
        \leq \frac{4 \parens*{Q(\mathbf{x}_1) - Q^*}}{\sqrt{\thresh T}} + \frac{4\sigma^2 L}{\sqrt{\thresh T}} + \parens*{\frac{8T}{\thresh}  + 8L^2} \max_{t\in [T]} \max_{i,j\in V_t} \E \bracks*{\norm*{\mathbf{x}_t^i - \mathbf{x}^j_t}_2^2}
    \end{equation*}
\end{lemma}

\begin{proof}
    Since $\eta = \frac{\sqrt{\thresh}}{\sqrt{T}}\leq \frac{1}{4L}$, following Lemma~\ref{non-convex-step} and \ref{effect-grad-diff}
    \begin{equation*}
        \frac{\eta}{4} \E\bracks*{\norm*{\nabla Q(\mathbf{x}_t^i)}_2^2}
        \leq \E\bracks*{Q(\mathbf{x}_{t}^i)} - \E \bracks*{Q(\mathbf{x}_{t+1}^i)} + \parens*{\frac{2}{\eta} + 2L^2\eta}\max_{i,j\in V_t} \E \bracks*{\norm*{\mathbf{x}_t^i - \mathbf{x}^j_t}_2^2} + \frac{\sigma^2\eta^2 L}{\thresh} 
    \end{equation*}
    Summing the previous inequality over $t$ yields
    \begin{multline*}
        \sum_{t=1}^T \frac{\eta}{4} \E\bracks*{\norm*{\nabla Q(\mathbf{x}_t^i)}_2^2} \\
        \leq \sum_{i=1}^T \parens*{\E\bracks*{Q(\mathbf{x}_{t}^i)} - \E \bracks*{Q(\mathbf{x}_{t+1}^i)}} + \sum_{t=1}^T \frac{\sigma^2\eta^2 L}{\thresh} + \sum_{t=1}^T \parens*{\frac{2}{\eta} + 2L^2\eta}\max_{i,j\in V_t} \E \bracks*{\norm*{\mathbf{x}_t^i - \mathbf{x}^j_t}_2^2} \\
        = Q(\mathbf{x}_1) - \E \bracks*{Q(\mathbf{x}_{T+1}^i)} + \frac{\sigma^2\eta^2 L}{\thresh} T + \parens*{\frac{2}{\eta} + 2L^2\eta} \sum_{t=1}^T \max_{i,j\in V_t} \E \bracks*{\norm*{\mathbf{x}_t^i - \mathbf{x}^j_t}_2^2}
    \end{multline*}
    Therefore, using that $\tau$ is drawn uniformly at random,
    \begin{align*}
        {}&\E \bracks*{\norm*{\nabla Q(\mathbf{x}_\tau^i)}_2^2} 
        = \frac{1}{T} \sum_{t=1}^T \E\bracks*{\norm*{\nabla Q(\mathbf{x}_t^i)}_2^2} \\
        {}&\leq \frac{4\parens*{Q(\mathbf{x}_1) - Q^*}}{T\eta} + \frac{4\sigma^2\eta L}{\thresh} + \parens*{\frac{8}{\eta^2} + 8L^2} \max_{t\in [T]} \max_{i,j\in V_t} \E \bracks*{\norm*{\mathbf{x}_t^i - \mathbf{x}^j_t}_2^2} \\
        {}&= \frac{4 \parens*{Q(\mathbf{x}_1) - Q^*}}{\sqrt{\thresh T}} + \frac{4\sigma^2 L}{\sqrt{\thresh T}} + \parens*{\frac{8T}{\thresh} + 8L^2} \max_{t\in [T]} \max_{i,j\in V_t} \E \bracks*{\norm*{\mathbf{x}_t^i - \mathbf{x}^j_t}_2^2} && \explain{$\eta = \frac{\sqrt{\thresh}}{\sqrt{T}}$}
\end{align*}
\end{proof}

Lemma~\ref{thm:non-convex} shows that external convergence depends 
on internal convergence. 
(Recall from Section~\ref{section:strongly convex}, that in the special 
case where the function is strongly convex, 
both external and internal convergence are achieved naturally.) 
The proof of this lemma only uses the \emph{convexity} property of MDAA 
and gives motivation for adding the \emph{contraction} property
to ensure this term will be sufficiently small.
\paragraph{\emph{\bf Internal Convergence.}}
The next lemma bounds the diameter of the learning parameters of iteration $t+1$ 
relative to the diameter of the previous iteration $t$. 
It is proved by first bounding the diameter after each process performs a 
local SGD step,
and then using the contraction property of MDAA.  

\begin{lemma}
\label{rec-model-diam}
    For every iteration $t\geq 1$,
    \begin{equation*}
        \max_{i,j\in V_t} \E\bracks*{\norm*{\mathbf{x}_{t+1}^i - \mathbf{x}_{t+1}^j}_2^2} 
        \leq q \parens*{2 + 2L^2\eta_t^2} \max_{i,j\in V_{t-1}}\E\bracks*{\norm*{\mathbf{x}_t^i - \mathbf{x}_t^j}_2^2} + q \frac{4\sigma^2\eta_t^2}{\thresh{}}
    \end{equation*}
\end{lemma}

\begin{proof}
Consider two processes $i,j\in V_t$. Let $S_1$ be the set of processes that was used to compute $\mathbf{g}^i_t$ in Line~\ref{lin:avg-grad}, and let $S_2$ be the set of processes that was used to compute $\mathbf{g}^j_t$ in Line~\ref{lin:avg-grad}.
\begin{align*}
    &\E_t \bracks*{\norm*{\mathbf{y}_t^i - \mathbf{y}_t^j}_2^2} \\
    &= \E_t \bracks*{\norm*{ \mathbf{x}_t^i - \eta_t \mathbf{g}_t^i - \parens*{\mathbf{x}_t^j - \eta_t \mathbf{g}_t^j}}_2^2} \\
    &= \norm*{\mathbf{x}_t^i - \E_t \bracks*{\eta_t \mathbf{g}_t^i} - \mathbf{x}_t^j + \E_t \bracks*{\eta_t \mathbf{g}_t^j}}_2^2  + \Var \bracks*{\eta_t \mathbf{g}_t^i - \eta_t \mathbf{g}_t^j} && \explain{by Proposition~\ref{prop:sep-mean-var}}\\
    &\leq \norm*{\mathbf{x}_t^i - \eta_t \E_t \bracks*{\mathbf{g}_t^i} - \mathbf{x}_t^j + \eta_t \E_t \bracks*{\mathbf{g}_t^j}}_2^2  + 2 \eta_t^2 \Var \bracks*{\mathbf{g}_t^i} + 2 \eta_t^2 \Var \bracks*{\mathbf{g}_t^j} && \explain{by Proposition~\ref{prop:var-diff}}\\
    &\leq \norm*{\mathbf{x}_t^i - \frac{\eta_t}{|S_1|} \sum_{k\in S_1} \nabla Q(\mathbf{x}_t^k) - \mathbf{x}_t^j + \frac{\eta_t}{|S_2|} \sum_{l\in S_2} \nabla Q(\mathbf{x}_t^l)}_2^2 + \frac{4\sigma^2\eta_t^2}{\thresh} && \explain{by Lemma~\ref{var-total-non-convex} and \eqref{unbiased-est}}
\end{align*}
Note that
\begin{align*}
    &\norm*{\mathbf{x}_t^i - \frac{\eta_t}{|S_1|} \sum_{k\in S_1} \nabla Q(\mathbf{x}_t^k) - \mathbf{x}_t^j + \frac{\eta_t}{|S_2|} \sum_{l\in S_2} \nabla Q(\mathbf{x}_t^l)}_2^2 \\
    &\leq 2\norm*{\mathbf{x}_t^i - \mathbf{x}_t^j}_2^2 + 2\eta_t^2\norm*{\frac{1}{|S_2|} \sum_{l\in S_2} \nabla Q(\mathbf{x}_t^l) - \frac{1}{|S_1|} \sum_{k\in S_1} \nabla Q(\mathbf{x}_t^k)}_2^2 && \explain{by \eqref{relaxed-triangle-ineq2}} \\
    &\leq 2\norm*{\mathbf{x}_t^i - \mathbf{x}_t^j}_2^2 + \frac{2\eta_t^2}{|S_1||S_2|} \sum_{k\in S_1} \sum_{l\in S_2} \norm*{\nabla Q(\mathbf{x}_t^l) - \nabla Q(\mathbf{x}_t^k)}_2^2 && \explain{by Proposition~\ref{prop:l2-avg}}\\
    &\leq 2\norm*{\mathbf{x}_t^i - \mathbf{x}_t^j}_2^2 + \frac{2\eta_t^2 L^2}{|S_1||S_2|} \sum_{k\in S_1} \sum_{l\in S_2} \norm*{\mathbf{x}_t^l - \mathbf{x}_t^k}_2^2 && \explain{by \eqref{l-lipsc}}
\end{align*}
Therefore,
\begin{align*}
    \E \bracks*{\norm*{\mathbf{y}_t^i - \mathbf{y}_t^j}_2^2} 
    &= \E\bracks*{\E_t \bracks*{\norm*{\mathbf{y}_t^i - \mathbf{y}_t^j}_2^2} }
    \leq \parens*{2 + 2\eta_t^2 L^2} \max_{k,l\in V_{t-1}}\E\bracks*{\norm*{\mathbf{x}_t^k - \mathbf{x}_t^l}_2^2} + \frac{4\sigma^2\eta_t^2}{\thresh}
\end{align*}
Since this is true for any pair of processes $i,j\in V_t$ we get that
\begin{equation*}
    \label{eq:3}
     \max_{i,j\in V_t} \E \bracks*{\norm*{\mathbf{y}_t^i - \mathbf{y}_t^j}_2^2} 
     \leq \parens*{2 + 2\eta_t^2 L^2} \max_{i,j\in V_{t-1}}\E\bracks*{\norm*{\mathbf{x}_t^i - \mathbf{x}_t^j}_2^2} + \frac{4\sigma^2\eta_t^2}{\thresh}
\end{equation*}
Finally, using $q$-contraction property of MDAA in Line~\ref{lin:approx-agree} implies the lemma.
\end{proof}

When $q \parens*{2 + 2L^2\eta_t^2}< 1$ we get contraction relative to 
the previous iteration, with an additive term. 
Assuming that $\eta_t \leq \frac{1}{L}$, yields that $2 + 2L^2\eta_t^2 \leq 4$. 
Hence, for any $q < \frac{1}{4}$ this term is smaller than 1. 
Since the additive term also depends on $q$, we can use it to control its magnitude.
The next lemma bounds the distance between the learning parameters at each iteration, 
using a constant learning rate $\eta$ and contraction parameter $q\approx \eta$. (The proof of the lemma is deferred to Appendix~\ref{app:add-proofs}.)

\begin{restatable}{lemma}{smalldiam}
\label{small-diam}
    Consider Algorithm~\ref{alg:dis-sgd} with constant learning rate 
    $\eta \leq \min \braces*{\frac{1}{2},\frac{1}{L}}$ and
    $q = \frac{\eta}{4}$. Then for every iteration $t\geq 1$,
    \begin{equation*}
        \max_{i,j\in V_t} \E\bracks*{\norm*{\mathbf{x}_{t+1}^i - \mathbf{x}_{t+1}^i}_2^2} \leq  \frac{2\sigma^2 \eta^3}{\thresh}
    \end{equation*}
\end{restatable}
Lemma~\ref{small-diam} implies that setting 
$\eta = \mathcal{O} \parens*{\sqrt[3]{\thresh \epsmodel / \sigma^{2}}}$ 
yields internal convergence. 
Using the same learning rate as in Lemma~\ref{thm:non-convex}, 
$\eta = \sqrt{\thresh{}/T}$, 
we get that $\max_{i,j\in V_{\tau-1}} \E\bracks*{\norm*{\mathbf{x}_{\tau}^i - \mathbf{x}_{\tau}^j}_2^2}
    = \mathcal{O} \parens*{\thresh^{1/2} T^{-{3}/{2}}}$.
This yields internal convergence in
${\mathcal{O}} \parens*{\thresh^{1/3}\epsmodel^{-2/3}}$ iterations.

\paragraph{\emph{\bf External Convergence.}}
Finally, by bounding the distance of the learning parameters at each iteration, 
we prove that Algorithm~\ref{alg:dis-sgd},
using an MDAA algorithm with $\mathcal{O}\parens*{\log T}$ communication rounds, 
has convergence rate that matches the sequential SGD algorithm, 
up to a logarithmic factor in the number of rounds. 
In a nutshell, the proof of the theorem uses the bound on the distance 
between the learning parameters from Lemma~\ref{small-diam} 
in Lemma~\ref{thm:non-convex}. 

\begin{lemma}
\label{lem:non-convex-small-error}
Let $Q$ be an $L$-smooth cost function. 
Consider Algorithm~\ref{alg:dis-sgd} with $T\geq \max\braces*{16L^2\thresh,4\thresh}$, constant learning rate $\eta = \frac{\sqrt{\thresh}}{\sqrt{T}}$ and parameters set as in Lemma~\ref{small-diam},
then for every process $i\in V_T$
    \begin{equation*}
        \E\bracks*{\norm*{\nabla Q(\mathbf{x}_\tau^i)}_2^2} 
        \leq \frac{4 \parens*{Q(\mathbf{x}_1) - Q^*}}{\sqrt{\thresh T}} + \frac{20\sigma^2 \max\braces*{L,L^2}}{\sqrt{\thresh T}} + \frac{16\sigma^2}{\sqrt{\thresh T}}
    \end{equation*}
\end{lemma}

\begin{proof}
    Note that $\eta = \frac{\sqrt{\thresh}}{\sqrt{T}}\leq \min\braces*{\frac{1}{4L},\frac{1}{2}}$. Following Lemma~\ref{thm:non-convex}
    \begin{align*}
        &\E \bracks*{\norm*{\nabla Q(\mathbf{x}_\tau^i)}_2^2} \\
        &\leq \frac{4 \parens*{Q(\mathbf{x}_1) - Q^*}}{\sqrt{\thresh T}} + \frac{4\sigma^2 L}{\sqrt{\thresh T}} + \parens*{\frac{8T}{\thresh}  + 8L^2} \max_{t\in [T]} \max_{i,j\in V_t} \E \bracks*{\norm*{\mathbf{x}_t^i - \mathbf{x}^j_t}_2^2} \\
        &\leq \frac{4 \parens*{Q(\mathbf{x}_1) - Q^*}}{\sqrt{\thresh T}} + \frac{4\sigma^2 L}{\sqrt{\thresh T}} + \parens*{\frac{16T}{\thresh} + 16L^2} \frac{\sigma^2 \eta^3}{\thresh} && \explain{by Lemma~\ref{small-diam}}\\
        &= \frac{4 \parens*{Q(\mathbf{x}_1) - Q^*}}{\sqrt{\thresh T}} + \frac{4\sigma^2 L}{\sqrt{\thresh T}} + \parens*{\frac{16 T}{\thresh} + 16L^2} \frac{\sqrt{\thresh}\sigma^2}{T^{3/2}} && \explain{$\eta = \frac{\sqrt{\thresh}}{\sqrt{T}}$}\\
         &= \frac{4 \parens*{Q(\mathbf{x}_1) - Q^*}}{\sqrt{\thresh T}} + \frac{4\sigma^2 L}{\sqrt{\thresh{}T}} + \frac{16\sigma^2}{\sqrt{\thresh{}T}} + \frac{16 \sqrt{\thresh} L^2\sigma^2}{T^{3/2}} \\
         &\leq \frac{4 \parens*{Q(\mathbf{x}_1) - Q^*}}{\sqrt{\thresh T}} + \frac{20\sigma^2 \max\braces*{L,L^2}}{\sqrt{\thresh{}T}} + \frac{16\sigma^2}{\sqrt{\thresh T}} && \explain{$T\geq \thresh$}
    \end{align*}
\end{proof}

As each iteration consists of several communication rounds, we finally get:
\begin{theorem}
\label{thm:non-convex-small-error}
Let $Q$ be an $L$-smooth cost function. 
Consider Algorithm~\ref{alg:dis-sgd} with $T\geq \max\braces*{16L^2\thresh,4\thresh}$, constant learning rate $\eta = \frac{\sqrt{\thresh}}{\sqrt{T}}$ and $q = \frac{\eta}{4}$. Then, after $R = \mathcal{O} \parens*{T  \log T}$ communication rounds for some constant $C$ and every process $i\in V_T$
    \begin{equation*}
        \E\bracks*{\norm*{\nabla Q(\mathbf{x}_\tau^i)}_2^2} 
        \leq \frac{C \parens*{Q(\mathbf{x}_1) - Q^*} \log R}{\sqrt{\thresh R}} + \frac{5C\sigma^2 \max\braces*{L,L^2} \log R}{\sqrt{\thresh R}} + \frac{4C\sigma^2 \log R}{\sqrt{\thresh R}}
    \end{equation*}
\end{theorem}

\begin{proof}
By Lemma~\ref{lem:non-convex-small-error}, 
    \begin{equation*}
        \E\bracks*{\norm*{\nabla Q(\mathbf{x}_\tau^i)}_2^2} 
        \leq \frac{4 \parens*{Q(\mathbf{x}_1) - Q^*}}{\sqrt{\thresh T}} + \frac{20\sigma^2 \max\braces*{L,L^2}}{\sqrt{\thresh T}} + \frac{16\sigma^2}{\sqrt{\thresh T}}.
    \end{equation*}
   By Theorem~\ref{thm:MDAA}, the MDAA algorithm requires $C\left\lceil \log \sqrt{T/\thresh} \right\rceil$ communication rounds at each iteration, for some constant $C$. After $T$ iterations, this yields $R = C T \left\lceil \log \parens*{\sqrt{T/\thresh}} \right\rceil \leq  C T \log_2 T$. Hence, using that $T\leq R$ and therefore, $\log T \leq \log R$
    \begin{align*}
        \E\bracks*{\norm*{\nabla Q(\mathbf{x}_\tau^i)}_2^2} 
        &\leq \frac{4C \parens*{Q(\mathbf{x}_1) - Q^*} \log R}{\sqrt{\thresh R}} + \frac{20C\sigma^2 \max\braces*{L,L^2} \log R}{\sqrt{\thresh R}} + \frac{16C\sigma^2 \log R}{\sqrt{\thresh R}}
    \end{align*}
\end{proof}


Theorem~\ref{thm:non-convex-small-error} shows that 
Algorithm~\ref{alg:dis-sgd} converges externally 
in $\widetilde{\mathcal{O}} \parens*{\epsconverge^{-2}/\thresh{}}$ communication rounds.
Recall that by Lemma~\ref{small-diam} the algorithm converges internally in 
$\widetilde{\mathcal{O}} \parens*{\thresh{}^{1/3}\delta^{-2/3}}$ communication rounds.
In both cases, dependencies on $\parens*{Q(\mathbf{x}_1) - Q^*}$, 
$L$ and $\sigma$ are neglected. 
In the non-convex case, internal convergence is faster than 
external convergence.

\section{Cluster-Based MDAA}
\label{sec:approximate agreement}

\begin{algorithm}[tb]
    \caption{Cluster-based multidimensional approximate agreement:
    code for process $i$ in cluster $c$}
    \label{alg:multi-approx-agree}
\begin{algorithmic}[1]
    \Statex {\sf MDAA}($\mathbf{x}, q$):
    \State $\mathbf{x}_1^i \gets \mathbf{x}$
    \For{$r=1\dots R = \lceil \log_{{23}/{24}} q \rceil$}
    \label{lin:loop}
        \State $\mathbf{y}_r^i \gets$ {\sf SMMDAA}$_r$($\mathbf{x}_r^i, 1/6$)
        \label{lin:sm-approx-agree}
        \Comment{Shared-memory algorithm inside cluster $c$}
        \State broadcast $\tup{r, \mathbf{y}_r^i}$ to all processes
        \State wait to receive messages of the form $\tup{r,-}$ \emph{representing} $n - \fp$ processes
        \State $Rcv \gets$ set of received values
        \label{lin:set}
        \State $\mathbf{x}_{r+1}^i \gets $ {\sf MidExtremes}($Rcv$)
        \label{lin:agg-rule}
    \EndFor
    \State output $\mathbf{x}_{R+1}^i$
\end{algorithmic}
\end{algorithm}
Algorithm~\ref{alg:multi-approx-agree} solves MDAA in the cluster-based model.
The algorithm leverages inter-cluster communication to increase the number of failures 
that can be tolerated, and only requires read and write operations.
The algorithm works in rounds, each starting with \emph{shared-memory} 
MDAA (SMMDAA) within each cluster (Line~\ref{lin:sm-approx-agree}).
This algorithm satisfies the two properties defined in Section~\ref{section:non convex} (convexity and contraction),
among the inputs of each cluster separately. Note that each round $1\leq r\leq R$ (and cluster) has a separate instance of the shared-memory algorithm.
The processes of each cluster communicate in this algorithm using only 
the common shared memory.
(See Algorithm~\ref{alg:sm-multi-approx-agree} 
in Section~\ref{sec:smmdaa} for full details and proof.)
This algorithm allows processes to wait to 
\emph{receive messages representing a majority of the processes at each round, 
rather than waiting for a majority of the processes}, 
which is the usual practice in crash-tolerant message-passing algorithms.
This guarantees that every pair of processes receive a value from a common cluster. 
Any two values sent from the same cluster have smaller diameter, 
compared to the diameter of all the process values, 
since they are the output of the inter-cluster SMMDAA algorithm. 
After collecting sufficiently many messages, 
the process computes the next round value using the \emph{MidExtremes}
rule~\cite{fggerFastApproximate}, 
which returns the average of the two values that realize 
the maximum Euclidean distance among all received vectors. 
Formally, for $X \subseteq \R^d$ 
\begin{equation*}
    {\sf MidExtremes}(X) = \parens*{\mathbf{a} + \mathbf{b}}/2\text{, where } (\mathbf{a},\mathbf{b}) = \argmax_{(\mathbf{a},\mathbf{b})\in X^2} \norm*{\mathbf{a} - \mathbf{b}}_2. 
\end{equation*}

We assume that $\fp\leq \fopt$, which by Lemma~\ref{lemm:fopt} implies that 
two sets of processes, each representing $n-f$ processes,
must include a process from the same cluster
(not necessarily the same process). 
This is 
optimal, since even \emph{one-dimensional} approximate agreement cannot be 
solved when $\fp > \fopt$~\cite{AttiyaKS20}.

Let $V_r$ be the processes that compute a value for round $r+1$ 
(Line~\ref{lin:agg-rule}), 
and let $V_0$ be all the processes that execute the first line. Let $\mathbf{x}_r$ be the multi-set of round $r$ values computed by processes in $V_{r-1}$ and $\mathbf{y}_r$ be the multi-set of the outputs from the SMMDAA algorithm in Line~\ref{lin:sm-approx-agree} of round $r$.
Note that at each round the value is a convex combination of the previous rounds' values, hence, this algorithm satisfies the convexity property.
Additionally, the convexity property of the shared-memory algorithm implies that
for any round $r$
\begin{equation}
\label{eq:4}
    \diameter{\mathbf{y}_r} \leq \diameter{\mathbf{x}_r}.
\end{equation}
As a reminder, the {diameter} of a set $A\subseteq \R^d$ is defined as
$\diameter{A} \triangleq \max_{\mathbf{x},\mathbf{y}\in A} \norm*{\mathbf{x} - \mathbf{y}}_2^2$.

In each round of our MDAA algorithm, 
the diameter of the values is reduced by a fixed \emph{contraction rate} 
relative to the diameter of the values in the previous round. 
The \emph{round-$r$ contraction rate} is an upper bound 
$\rho_r \geq {\diameter{\mathbf{x}_{r+1}}}/{\diameter{\mathbf{x}_r}}$.
We consider algorithms with a uniform upper bound $\rho<1$ 
on the contraction rate of all rounds. 
Contraction rate $\rho$ ensures $q$-contraction within 
$\lceil \log_{\rho} q \rceil \approx \log_2 (1/q)$ rounds (assuming $q, \rho < 1$).

We explain our cluster-based MDAA algorithm in the one-dimensional case; 
the algorithm and full proof cover the multidimensional case.
Intuitively, rather than requiring that collected sets intersect,
it suffices to assume that they contain ``close enough" values.
When $d=1$, MidExtremes returns the \emph{MidPoint},
i.e., for a set $X\subseteq \R$, ${\sf midpoint}(X) = \parens*{\min(X) + \max(X)}/2$. 
We also have that $\diameter{X} = \max(X) - \min(X)$ for $X\subseteq \R$.
Let $A,B \subseteq U \subseteq \R$, $A$ and $B$ stand for values collected in the same round by two different processes (Line~\ref{lin:set}), and $U$ is all the possible round values. For any pair of values $a\in A$ and $b\in B$, 
we have ${\sf midpoint}(A) \leq \parens*{a + \max(U)}/2$ 
and ${\sf midpoint}(B) \geq \parens*{\min(U) + b}/2$, which implies:
\begin{equation}
\label{eq:MP}
    {\sf midpoint}(A) - {\sf midpoint}(B) 
    \leq \frac{1}{2} \underbrace{\parens*{\max(U) - \min(U)}}_{\diameter{U}} + \frac{1}{2} \parens*{a - b}
\end{equation}

When $A$ and $B$ intersect, the second term in \eqref{eq:MP} zeros out, which gives contraction rate of $1/2$ when any two sets of collected 
values intersect. (In the general $d$-dimensional case, MidExtremes has contraction rate of $7/8$ under this assumption~\cite{fggerFastApproximate}.) The assumption that any two collected sets of size $n-f$ must intersect is not guaranteed when only assuming $\fp \leq \fopt$. Instead, we can use an SMMDAA algorithm with contraction parameter of 1/2. Assuming there are values $a\in A$ and $b\in B$ such that 
$a - b \leq \diameter{U}/2$, by~\eqref{eq:MP}, 
$\abs*{{\sf midpoint}(A) - {\sf midpoint}(B)} \leq 3/4\,\diameter{U}$,
leading to a slightly worse contraction rate of $3/4$.

Lemma~\ref{additive-contraction} generalizes
\cite[Lemmas 4 and 5]{fggerFastApproximate}
and is proved using the next lemma.
\begin{lemma}[{\cite[Lemma 3]{fggerFastApproximate}}]
\label{fugger-lemma3}
    Let $\mathbf{a},\mathbf{b},\mathbf{c}\in \R^d$, then setting $\mathbf{m} = \parens*{\mathbf{a} + \mathbf{b}}/2$
    \begin{equation*}
        \norm*{\mathbf{m} - \mathbf{c}}_2^2 
        \leq \frac{1}{2} \norm*{\mathbf{a} - \mathbf{c}}_2^2 + \frac{1}{2} \norm*{\mathbf{b} - \mathbf{c}}_2^2 -\frac{1}{4} \norm*{\mathbf{a} - \mathbf{b}}_2^2
    \end{equation*}
\end{lemma}

\begin{lemma}
\label{additive-contraction}
    Let $\mathbf{a}, \mathbf{a}', \mathbf{b}, \mathbf{b}'\in \R^d$ such that
    \begin{equation}
        \label{eq:7}
        \diameter{\braces*{\mathbf{a}, \mathbf{a}', \mathbf{b}, \mathbf{b}'}} \leq c\norm*{\mathbf{a} - \mathbf{b}}_2^2 + c\norm*{\mathbf{a}' - \mathbf{b}'}_2^2 + d.
    \end{equation}
    for $c,d>0$. Then, setting $\mathbf{m} = \parens*{\mathbf{a} - \mathbf{b}}/2$ and $\mathbf{m}' = \parens*{\mathbf{a}' - \mathbf{b}'}/2$
    \begin{equation*}
        \norm*{\mathbf{m} - \mathbf{m'}}_2^2 
        \leq \frac{4c-1}{4c}\diameter{\braces*{\mathbf{a},\mathbf{a}',\mathbf{b},\mathbf{b}'}} + \frac{d}{4c}
    \end{equation*}
\end{lemma}

\begin{proof}
    Denote $\delta =  \diameter{\braces*{\mathbf{a},\mathbf{a}',\mathbf{b},\mathbf{b}'}}$. By Lemma~\ref{fugger-lemma3}
\begin{equation}
\begin{aligned}
\label{eq:2}
    \norm*{\mathbf{a} - \mathbf{m'}}_2^2 
    &= \frac{1}{2} \norm*{\mathbf{a}' - \mathbf{a}}_2^2 + \frac{1}{2} \norm*{\mathbf{b}' - \mathbf{a}}_2^2 - \frac{1}{4} \norm*{\mathbf{a}' - \mathbf{b}'}_2^2 
    \leq \delta - \frac{1}{4} \norm*{\mathbf{a}' - \mathbf{b}'}_2^2 \\
    \norm*{\mathbf{b} - \mathbf{m'}}_2^2 
    &= \frac{1}{2} \norm*{\mathbf{a}' - \mathbf{b}}_2^2 + \frac{1}{2} \norm*{\mathbf{b}' - \mathbf{b}}_2^2 - \frac{1}{4} \norm*{\mathbf{a}' - \mathbf{b}'}_2^2 
    \leq \delta - \frac{1}{4} \norm*{\mathbf{a}' - \mathbf{b}'}_2^2 
\end{aligned}
\end{equation}
Hence,
\begin{align*}
    \norm*{\mathbf{m} - \mathbf{m'}}_2^2 
    &= \frac{1}{2} \norm*{\mathbf{a} - \mathbf{m'}}_2^2 + \frac{1}{2} \norm*{\mathbf{b} - \mathbf{m'}}_2^2 -\frac{1}{4} \norm*{\mathbf{a} - \mathbf{b}}_2^2 && \explain{by Lemma~\ref{fugger-lemma3}}\\
    &\leq \delta - \frac{1}{4} \parens*{\norm*{\mathbf{a} - \mathbf{b}}_2^2 + \norm*{\mathbf{a}' - \mathbf{b}'}_2^2} && \explain{by \eqref{eq:2}}\\
    &\leq \delta - \frac{1}{4c} \delta +\frac{d}{4c} && \explain{by \eqref{eq:7}}\\
    &= \frac{4c-1}{4c}\delta + \frac{d}{4c}
\end{align*}
\end{proof}

The next lemma relates 
the SMMDAA contraction parameter $q_{sm}$ 
and the global contraction rate $\rho$ of the outer algorithm in the general $d$-dimensional case, using Lemma~\ref{additive-contraction},
which is $\rho \leq {11}/{12} + {q_{sm}}/{4}$. 
Using $q_{sm} = 1/6$, we get an MDAA algorithm with contraction rate $23/24$.

\begin{lemma}
\label{approx-agree-round-contraction}
For any round $r\geq 1$, 
    $\diameter{\mathbf{x}_{r+1}} \leq \frac{23}{24} \diameter{\mathbf{x}_{r}}$.
\end{lemma}

\begin{proof}
Consider two processes $i,j\in V_r$. Let $\mathbf{a}$, $\mathbf{b}$ be the values such that $\mathbf{x}_{r+1}^i = \parens*{\mathbf{a} + \mathbf{b}}/2$ and $\mathbf{a}'$, $\mathbf{b}'$ be the values such that $\mathbf{x}_{r+1}^j = \parens*{\mathbf{a}' + \mathbf{b}'}/2$. Since processes $i$ and $j$ received messages representing $n-f$ processes and $f\leq \fopt$, following Lemma~\ref{lemm:fopt}, both processes must have received messages from the same cluster $c$. Thus, process $i$ received value $\mathbf{y}_r^k$ and process $j$ received value $\mathbf{y}_r^l$, where $k$ and $l$ are part the same cluster. 
Since both $\mathbf{y}_r^k$ and $\mathbf{y}_r^l$ are the outputs of the same SMMDAA algorithm satisfying $1/6$-contraction, $\norm*{\mathbf{y}_r^k - \mathbf{y}_r^l}_2^2 = d \leq \frac{1}{6}\,\diameter{\mathbf{x}_r}$.
As MidExtremes returns the values realizing the maximal Euclidean distance we have that $\norm*{\mathbf{a} - \mathbf{y}_r^k}_2^2 \leq \norm*{\mathbf{a} - \mathbf{b}}_2^2$ and $\norm*{\mathbf{a}' - \mathbf{y}_r^l}_2^2 \leq \norm*{\mathbf{a}' - \mathbf{b}'}_2^2$.
Hence, By using \eqref{relaxed-triangle-ineq2}, 
\begin{align*}
    \norm*{\mathbf{a} - \mathbf{a}'}_2^2 &\leq 3\norm*{\mathbf{a} - \mathbf{y}_r^k}_2^2 + 3\norm*{\mathbf{y}_r^k - \mathbf{y}_r^l}_2^2 + 3\norm*{\mathbf{y}_r^l - \mathbf{a'}}_2^2 \\
    &\leq 3\norm*{\mathbf{a} - \mathbf{b}}_2^2 + 3\norm*{\mathbf{a}' - \mathbf{b}'}_2^2 + 3d.
\end{align*}

Similarly, this holds for any pair in $\braces{\mathbf{a},\mathbf{a}',\mathbf{b},\mathbf{b}'}$, therefore,
 $\diameter{\braces*{\mathbf{a},\mathbf{a}',\mathbf{b},\mathbf{b}'}}
\leq 3\norm*{\mathbf{a} - \mathbf{b}}_2^2 + 3\norm*{\mathbf{a}' - \mathbf{b}'}_2^2 + 3d$. By Lemma~\ref{additive-contraction}
\begin{align*}
    \norm*{\mathbf{x}_{r+1}^i - \mathbf{x}_{r+1}^j}_2^2 
    &\leq \frac{11}{12} \diameter{\braces*{\mathbf{a},\mathbf{a}',\mathbf{b},\mathbf{b}'}} + \frac{d}{4} \\
    &\leq \frac{11}{12} \diameter{\mathbf{y}_r} + \frac{d}{4} \\
    &\leq \frac{11}{12} \diameter{\mathbf{x}_r} + \frac{1}{24} \diameter{\mathbf{x}_r} 
    = \frac{23}{24} \diameter{\mathbf{x}_r}  && \explain{by \eqref{eq:4}} 
\end{align*}
The lemma follows since the bound holds for 
every pair of vectors in $\mathbf{x}_{r+1}$.
\end{proof}

Lemma~\ref{approx-agree-round-contraction} shows that algorithm~\ref{alg:multi-approx-agree} achieves contraction rate of 23/24, hence we get our main result:
\begin{theorem}
\label{thm:MDAA}
Algorithm~\ref{alg:multi-approx-agree} satisfies $q$-contraction
within $\lceil \log_{{23}/{24}} q \rceil$ 
communication rounds.
\end{theorem}

\begin{proof}
    Following Lemma~\ref{approx-agree-round-contraction}, for any round $r\geq 1$, 
    $\diameter{\mathbf{x}_{r+1}} \leq \frac{23}{24} \diameter{\mathbf{x}_{r}}$ Inductivly, after $R=\lceil \log_{{23}/{24}} q \rceil$ rounds
    \begin{equation*}
        \diameter{\mathbf{x}_{R+1}} 
        \leq \parens*{\frac{23}{24}}^R \diameter{\mathbf{x}_1} 
        = \parens*{\frac{23}{24}}^{\lceil \log_{{23}/{24}} q \rceil} \diameter{\mathbf{x}_1}
        \leq q\, \diameter{\mathbf{x}_1}.
    \end{equation*}
\end{proof}

\subsection{Shared-Memory MDAA}
\label{sec:smmdaa}

\begin{algorithm}[tb]
    \caption{Shared-memory multidimensional approximate agreement: code for process $i$ in cluster $P_c$}
    \label{alg:sm-multi-approx-agree}
	\begin{algorithmic}[1]
	    \Statex \textbf{Shared variables:} $\forall i \in P_c$, $A[i]$, initially $\tup*{\bot,\bot}$
	    \Statex {\sf SMMDAA}($\mathbf{x}, q$):
	    \State $r \gets 1$
	    \State $A[i] \gets \tup*{1,\mathbf{x}}$
        \label{lin:2}
		\While{$r \leq R = \lceil \log_{7/8} q \rceil$}
		\State $\forall j\in P_c$, $\tup*{r_j,\mathbf{x}_j} \gets A[j]$
		\Comment{Collect $A$}
		\State $r_{max} \gets \max_{j\in P_c} r_j$
        \Comment{Ignore $\bot$-values}
		\If{$r = r_{max}$}
		\State $X \gets \braces*{\mathbf{x}_j\,|\,r_j = r_{max}}$
		\label{lin:after-collect}
		    \State $A[i] \gets \tup*{r+1,\text{{\sf MidExtremes}}(X)}$
                \label{lin:write-val}
		    \State $r \gets r+1$
		\Else{} 
		    \State $r \gets r_{max}$
		    \label{lin:skip}
                \Comment{Skip to round $r_{max}$}
		\EndIf
		\label{lin:sm-agg-rule}
        \EndWhile
        \State \algorithmicreturn{} random $\mathbf{x}_j$ s.t. $\tup*{R+1,\mathbf{x}_j}= A[j]$
	\end{algorithmic}
\end{algorithm}

Algorithm~\ref{alg:sm-multi-approx-agree} is executed inside cluster $P_c$, i.e., among only the processes in the cluster which communicate using the shared memory. 
It adapts the algorithm of~\cite{fggerFastApproximate},
originally for the message-passing model, to the shared-memory model.
The algorithm is similar to the simple shared-memory snapshot 
algorithm~\cite{Moran95} (see~\cite{HagitBook}), 
but it uses only atomic read and writes to the shared single-writer 
multi-reader registers. 
To reuse the same array for all rounds, the algorithm employs \emph{skipping}, 
where a ``slow'' process joins the most recent round.

The algorithm proceeds in (asynchronous) rounds: 
in each round, the process reads all the values written by other processes for this round,
and aggregates to compute its value for the next round. 
To ensure a common value for two sets of values collected 
at the same round, a process does not write to $A$ if it reads
a later round number from $A$ then its current round (i.e., its current $r$ value). 
If this happens, the process skips to the maximal round number it has read.


The next lemma is used to bound the convergence rate of MidExtremes
when every pair of processes collects a common value in each round.
Note that $\sqrt{\diameter{A}} = \max_{\mathbf{x},\mathbf{y}\in A} \norm*{\mathbf{x} - \mathbf{y}}_2$.

\begin{lemma}[{\cite[Lemma 4]{fggerFastApproximate}}]
\label{fugger-lemma4}
    Let $\mathbf{a}, \mathbf{a}', \mathbf{b}, \mathbf{b}'\in \R^d$ such that $\sqrt{\diameter{\braces*{\mathbf{a}, \mathbf{a}', \mathbf{b}, \mathbf{b}'}}} \leq \norm*{\mathbf{a} - \mathbf{b}}_2 + \norm*{\mathbf{a}' - \mathbf{b}'}_2$. Then, setting $\mathbf{m} = \parens*{\mathbf{a} - \mathbf{b}}/2$ and $\mathbf{m}' = \parens*{\mathbf{a}' - \mathbf{b}'}/2$, 
    \begin{equation*}
        \norm*{\mathbf{m} - \mathbf{m'}}_2^2 \leq \frac{7}{8} \diameter{\braces*{\mathbf{a}, \mathbf{a}', \mathbf{b}, \mathbf{b}'}}
    \end{equation*}
\end{lemma}

Let $V_r$, $r\geq 1$, be the set of processes that write a value with round number $r+1$ to $A$ (i.e., $\tup*{r+1,-}$) in Line~\ref{lin:write-val}, by the code this must happen in round $r$. Let $V_0$ be the set of processes that write a value with round number $1$ to $A$ in Line~\ref{lin:2}.
Let $\mathbf{x}_{r}$ be the multi-set of all values ever written with round number $r$.
The next two properties can be easily proven using an inductive claim:
\begin{enumerate}
    \item For every $i\in P_c$, the value of $A[i]$ is strictly monotonically increasing, with respect to the written round number, during the execution of the algorithm.
    \item For every $0\leq r\leq R$, $V_r\neq \emptyset$.
\end{enumerate}

\begin{lemma}
\label{smmaa-contraction}
Algorithm~\ref{alg:sm-multi-approx-agree} satisfies $q$-contraction
within $\lceil \log_{{7}/{8}} q \rceil$ rounds.
\end{lemma}

\begin{proof}
Let round $1\leq r\leq R$ and let $i\in V_{r-1}$ be the first process to write to $A$ with round number $r$. Let $\mathbf{c}$ be the value written to $A[i]$ with round number $r$. Consider any process $j\in V_{r}$. Since $j$ writes a value with round number $r+1$, it must read $A$ in round $r$ and it holds that $r_{max} = r$ in that round.
We show that $\mathbf{c}$ must be part of the set  $X$ held by $j$ after executing Line~\ref{lin:after-collect} in round $r$.

There are two cases. In the first case $j\in V_{r-1}$, so $j$ writes to $A[j]$ with round number $r$. From our assumption, this write happens after the write to $A[i]$ with round number $r$. Since processes first write to $A$ with round number $r$ in round $r-1$ and then read $A$ in round number $r$, this implies that $j$ reads $A[i]$ in round $r$ after $i$ writes $\tup*{r,\mathbf{c}}$ to it.
Since the round numbers written to $A[i]$ only increase, $j$ must have read this value, otherwise, it must read a larger round number. By the code, the rounds processes take part in only increase during the execution, therefore, this implies that $j$ never writes with round number $r+1$.

In the second case, $j\notin V_{r-1}$. This implies that $j$ read round number $r$ from $A$ in a round $r'<r$, and then skipped to round $r$ through Line~\ref{lin:skip}. Thus, the read of $A$ by $j$ in round $r$ starts after $i$ writes $\tup*{r,\mathbf{c}}$ to $A[i]$. The rest of the proof is similar to the previous case.

Let two processes $i,j\in V_r$. Let $X_i$, $X_j$ be the value of $X$ held by $i$, $j$ after executing Line~\ref{lin:after-collect} in round $r$, respectively. By what we just proved, there is some value $\mathbf{c}$ such that $\mathbf{c}\in X_i\cap X_j$.
Let $\mathbf{a}$, $\mathbf{b}$ be the values such that $\mathbf{x}^i = \parens*{\mathbf{a} + \mathbf{b}}/2$ is written with round $r+1$ to $A[i]$ and $\mathbf{a}'$, $\mathbf{b}'$ be the values such that $\mathbf{x}^j\ = \parens*{\mathbf{a}' + \mathbf{b}'}/2$ is written with round $r+1$ to $A[j]$. Note that $\mathbf{a}, \mathbf{a}', \mathbf{b}$ and $\mathbf{b}'$ are part of $\mathbf{x}_r$.

Since $\mathbf{c}\in X_i$ we have that $\norm*{\mathbf{a} - \mathbf{c}}_2 \leq \norm*{\mathbf{a} - \mathbf{b}}_2$ and $\norm*{\mathbf{b} - \mathbf{c}}_2 \leq \norm*{\mathbf{a} - \mathbf{b}}_2$. 
Similarly,  as $\mathbf{c}\in X_j$, $\norm*{\mathbf{a}' - \mathbf{c}}_2 \leq \norm*{\mathbf{a}' - \mathbf{b}'}_2$ and $\norm*{\mathbf{b}' - \mathbf{c}}_2 \leq \norm*{\mathbf{a}' - \mathbf{b}'}_2$. Therefore, using triangle inequality, $\sqrt{\diameter{\braces*{\mathbf{a},\mathbf{a'},\mathbf{b},\mathbf{b'}}}} \leq \norm*{\mathbf{a} - \mathbf{b}}_2 + \norm*{\mathbf{a'} - \mathbf{b'}}_2$ and following Lemma~\ref{fugger-lemma4},
\begin{equation*}
    \norm*{\mathbf{x}^i - \mathbf{x}^j}_2^2 
    \leq \frac{7}{8} \diameter{\braces*{\mathbf{a}, \mathbf{a}', \mathbf{b}, \mathbf{b}'}}
    \leq \frac{7}{8} \diameter{\mathbf{x}_r}
\end{equation*}
Since this is true for any pair of processes in $V_r$, then $\diameter{\mathbf{x}_{r+1}} \leq \frac{7}{8} \diameter{\mathbf{x}_r}$.
After $R=\lceil \log_{{7}/{8}} q \rceil$ rounds
    \begin{equation*}
        \diameter{\mathbf{x}_{R+1}} 
        \leq \parens*{\frac{7}{8}}^R \diameter{\mathbf{x}_1} 
        = \parens*{\frac{7}{8}}^{\lceil \log_{{7}/{8}} q \rceil} \diameter{\mathbf{x}_1}
        \leq q\, \diameter{\mathbf{x}_1}.
    \end{equation*}
\end{proof}
By Lemma~\ref{smmaa-contraction}, Algorithm~\ref{alg:sm-multi-approx-agree} {provides} 
$1/6$-contraction in $14$ rounds.
\section{Impossibility of Asynchronous SGD with System Partitions}
\label{section:impossibility}

A non-convex function $Q$ may have several \emph{stationary points}, 
i.e., $\mathbf{x}\in \R^d$ such that $\nabla Q(\mathbf{x}) = 0$.
A stationary point can be either global or local minimum or maximum, 
or a \emph{saddle point}, which is not a local extremum of the function. 
Although SGD converges to a stationary point (see \eqref{convergence-require}),
it may happen that starting from the same initial point, 
SGD converges to different stationary points in different random executions.
If these stationary points are $\gamma$ apart, for a large enough $\gamma$,
and the probability of reaching them is large enough, 
then we can prove that no distributed SGD implementation, 
satisfying both internal and external convergence, 
can tolerate \emph{partitioning}. 
(Recall that in our model, the system is \emph{partitioned} if $\fp > \fopt$, 
and the correct processes may represent a minority of the processes.)
Intuitively, this is because each partition can converge individually 
to points that are $\gamma$ apart from each other, violating internal convergence.

Definition~\ref{def:split}, below, formalizes the phenomenon 
where the SGD algorithm converges $\gamma$ apart in two different 
random executions with probability at least $p$.

\paragraph{\emph{\bf Sequential Algorithm.}}
Let $\mathcal{A}^{seq}(Q,G,\mathbf{x}_0,\eta,T)$ be a sequential SGD algorithm optimizing 
the function $Q$ using stochastic gradient $G$ for $T$ iterations, starting from $\mathbf{x}_0$ with learning rate $\eta$. The stochastic gradient $G$ is an unbiased estimator of the gradient \eqref{unbiased-est} and has bounded variance \eqref{bounded-var}.
For simplicity, we assume a constant learning rate, but this definition can be extended to also consider decreasing learning rates.
Let $\mathbf{x}_{seq}$ be a random variable, 
corresponding to the output of the sequential algorithm. 
We denote the event that SGD converges to a point in a set of stationary points 
$S$ using the event $E_{seq}(\beta,S)$, for some small $\beta$. 
Formally, for $\beta\in \R$ and point $\mathbf{x}\in \R^d$, 
let $E_{seq}(\beta,\mathbf{x})$ be the event that 
$\norm*{\mathbf{x}_{seq} - \mathbf{x}}_2 \leq \beta$. 
For a set of points $S\subseteq \R^d$, 
$E_{seq}(\beta,S)$ denotes the event that for some $\mathbf{x}\in S$,
$\norm*{\mathbf{x}_{seq} - \mathbf{x}}_2 \leq \beta$,
i.e., $E_{seq}(\beta,S) = \bigcup_{\mathbf{x}\in S} E_{seq}(\beta,\mathbf{x})$.
Note that for a stationary point $\mathbf{x}$, using smoothness
\eqref{l-lipsc}, 
$\norm*{\nabla Q(\mathbf{x}_{seq})}_2^2
\leq L^2 \norm*{\mathbf{x}_{seq} - \mathbf{x}}_2^2$.
This implies that if the algorithm converges near a stationary point, 
then this point has small squared gradient.
The distance between two sets of points $S_1,S_2\subseteq \R^d$ is defined as
$\mathrm{dist}\parens*{S_1,S_2} \triangleq \min_{\mathbf{x}_1\in S_1, 
\mathbf{x}_2\in S_2} \norm*{\mathbf{x}_1 - \mathbf{x}_2}_2$, i.e., the smallest distance between a point in $S_1$ and a point in $S_2$.

\begin{definition}\label{def:split}
Function $Q$ is \emph{($\gamma$,$p$)-split} for $\gamma>0$ and $p\in (0,1)$, 
if there are two sets of points $S_1, S_2\subseteq \R^d$ where 
$\mathrm{dist}(S_1,S_2) \geq d$ and for some $\alpha,\beta\geq0$
such that $\parens*{d - \alpha - \beta} \geq \gamma$, 
$\,\Prob \bracks*{E_{seq}(\alpha,S_1)} \Prob \bracks*{E_{seq}(\beta,S_2)} \geq p$.
\end{definition}

Next, we present two examples of ($\gamma$,$p$)-split functions where $p = 1/2$.

\paragraph{\emph{\bf Example 1.}}
Consider the function $Q(x) = -x^2 e^{-x^2}$.
Its derivative is $\nabla Q(x) = Q'(x) = 2x\left(x^{2}-1\right)e^{-x^{2}}$.
The function $Q$ and its derivative are depicted in Figure~\ref{fig:example1}.
The function $Q$ is $L$-smooth (for some $L$) since its second derivative is bounded. 
There are two (global) minimum points, at $x\in \braces*{-1,1}$, both with the minimum value $-e^{-1}$.
Consider the following stochastic gradient; for every $x\neq 0$, $G(x) = \nabla Q(x)$, and for $x=0$, 
\begin{equation*}
    G(0) = 
\begin{cases}
    \sigma & \text{with probability 1/2} \\
    -\sigma & \text{with probability 1/2}
\end{cases}
\end{equation*}
Since $\nabla Q(0) = 0$, $\E\bracks*{G(x)} = \nabla Q(x)$ and $\E\bracks*{\parens*{G(x) -\nabla Q(x)}^2} \leq \sigma^2$.

Consider a sequential SGD algorithm $\mathcal{A}^{seq}$ with batch size $1$ optimizing $Q$ using $G$, starting at $x_1 = 0$ and with a constant step size $\eta$. At the first iteration, with probability 1/2 the direction of the gradient is positive, and with probability 1/2 it is negative. This translates to $x_2 = -\eta \sigma$ or $x_2 = \eta \sigma$, both with probability 1/2.
Once the direction is set, the algorithm becomes deterministic, since there is no randomness in the stochastic gradient. The gradient is {odd}, 
since $-Q'(x) = 2x\left(x^{2}-1\right)e^{-x^{2}} = Q'(-x)$,
therefore, the two versions of the SGD iterates are also odd.
Hence, if the first stochastic gradient is positive the algorithm will eventually converge to $x=1$, and if it is negative the algorithm will eventually converge to $x=-1$.
(This can be easily seen looking at Figure~\ref{fig:example1}.)
Once the parameter is set to one of these points ($x=1$ or $x=-1$), since the gradient is equal to 0 in both points, the parameters will never change again. 

To conclude, there is some iteration $T'$ such that for any $T\geq T'$, $\Prob \bracks*{E_{seq}(0,\{1\})} = \Prob \bracks*{E_{seq}(0,\{-1\})} = 1/2$.
This implies that with probability 1/4 given two independent executions of $\mathcal{A}^{seq}$, one will converge to $x=1$ and the other to $x=-1$.
Since the distance between 1 and -1 is 2, Following Definition~\ref{def:split}, $Q$ is $(2,1/4)$-split.

\paragraph{\emph{\bf Example 2.}}
Consider the function
$Q(x) = e^{-x^2}$.
Its derivative is $\nabla Q(x) = -2x e^{-x^2}$.
This function is $L$-smooth (for some $L$) since its second derivative is bounded. Both $Q$ 
and its derivative are depicted in Figure~\ref{fig:example2}.
Similarly to the previous example, consider the following stochastic gradient; for every $x\neq 0$, $G(x) = \nabla Q(x)$, and for $x=0$, $G(0) = \sigma$ w.p. 1/2 and $G(0) = -\sigma$ w.p. 1/2. Obviously, since $\nabla Q(x) = 0$, $\E\bracks*{G(x)} = \nabla Q(x)$ and $\E\bracks*{\parens*{G(x) -\nabla Q(x)}^2} \leq \sigma^2$.

Consider a sequential SGD algorithm $\mathcal{A}^{seq}$ with batch size $1$ optimizing $Q$ using $G$, starting at $x_1 = 0$, and consider the first iteration. With probability 1/2 the direction of the gradient is positive, and with probability 1/2 it is negative. Once the direction is set, the algorithm becomes deterministic, since there is no randomness in the stochastic gradient. The direction of all consecutive iterations stay positive or negative, depending on the direction of the first iteration. Since the value of $Q$ continues to decrease as $x$ goes to $\infty$ and $-\infty$,
for every $\gamma$ there is iteration $T$ such that $Q$ is $(\gamma,1/4)$-split.

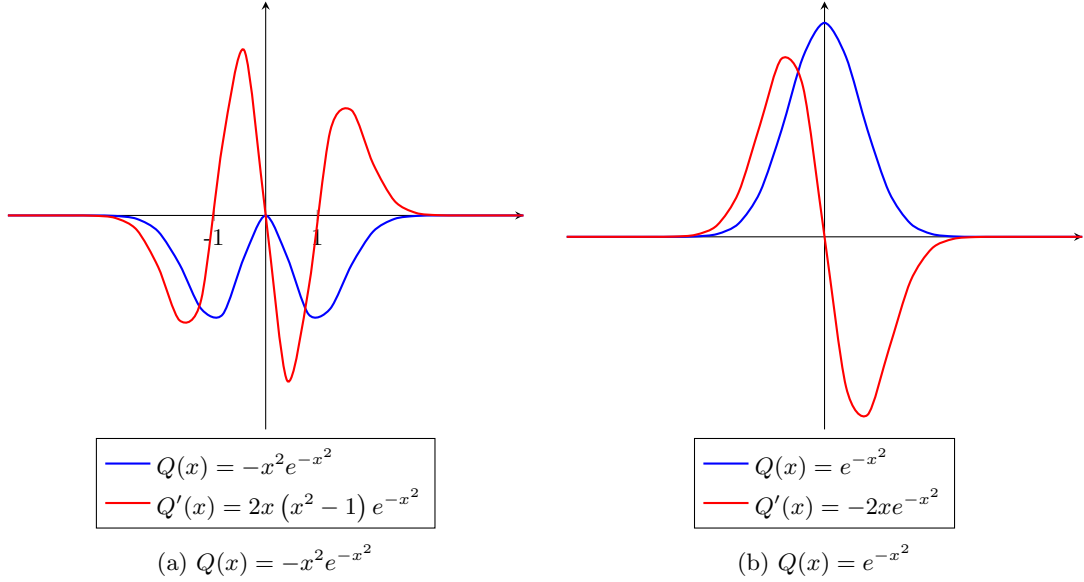
\begin{figure}[htp]
\centering
\begin{subfigure}{.45\textwidth}
  \centering
  \begin{tikzpicture}
  \begin{axis} [axis lines=center,   
    ymin=-0.75,ymax=0.75,
    ytick={0,0}, xtick={1,-1},
    xticklabels={1,-1},
    legend style={at={(0.5,-0.02)},anchor=north,legend cell align=left}
    ]
    \addplot [smooth, thick, blue] { -x^2 * exp(-x^2) };
    \addplot [smooth, thick, red] { 2*x * (x^2 - 1) * exp(-x^2) };
    \addlegendentry{$Q(x) = -x^2 e^{-x^2}$}
    \addlegendentry{$Q'(x) = 2x\left(x^{2}-1\right)e^{-x^{2}}$}
  \end{axis}
\end{tikzpicture}
  \caption{$Q(x) = -x^2 e^{-x^2}$}
  \label{fig:example1}
\end{subfigure}%
\begin{subfigure}{.45\textwidth}
  \centering
  \begin{tikzpicture}
  \begin{axis} [axis lines=center,   
    ymin=-0.9,ymax=1.1,
    ticks=none,
    legend style={at={(0.5,-0.02)},anchor=north,legend cell align=left}
    ]
    \addplot [smooth, thick, blue] { exp(-x^2) };
    \addplot [smooth, thick, red] { -2*x * exp(-x^2) };
    \addlegendentry{$Q(x) = e^{-x^2}$}
    \addlegendentry{$Q'(x) = -2xe^{-x^{2}}$}
  \end{axis}
\end{tikzpicture}
  \caption{$Q(x) = e^{-x^2}$}
  \label{fig:example2}
\end{subfigure}
\caption{Plot of $Q(x)$ and its gradient function}
\end{figure}

\paragraph{\emph{\bf Distributed Algorithm.}} 
{Consider an algorithm $\mathcal{A}$ and let $\mathcal{T}$ be its execution tree}. 
Let $\mathbf{x}^i(\mathcal{T})\in \R^d$ be a random variable, 
corresponding to the output of process $i$ from algorithm $\mathcal{A}$ 
over execution tree $\mathcal{T}$.
If $i$ crashed in the execution then this value is $\bot$. 
Let $E(\beta,\mathbf{x},\mathcal{T},i)$ 
be the event that for process $i$ such that
$\mathbf{x}^i(\mathcal{T})\neq \bot$,
$\norm*{\mathbf{x}^i(\mathcal{T}) - \mathbf{x}}_2 \leq \beta$. 
For a set of points $S$, 
let $E(\beta,S,\mathcal{T},i) = \bigcup_{\mathbf{x}\in S} E(\beta,\mathbf{x},\mathcal{T},i)$.

The next lemma shows that if two processes can converge near points that are far enough from each other with high enough probability, then their outputs will also be far from each other in expectation.
\begin{lemma}
\label{impossibility-helper}
    Let $\mathcal{A}$ be a cluster-based algorithm, $\mathcal{T}$ an execution tree of $\mathcal{A}$ and two processes $i$ and $j$. If $\Prob \bracks*{E(\alpha,S_1,\mathcal{T},i) \cap E(\beta,S_2,\mathcal{T},j)} \geq p$ and $\Prob \bracks*{E(\alpha,S_1,\mathcal{T},j) \cap E(\beta,S_2,\mathcal{T},i)} \geq p$ for $\alpha,\beta \geq 0$ and $S_1,S_2\subseteq \R^d$ such that $\mathrm{dist}\parens*{S_1,S_2} \geq d > \max\braces*{0,\alpha+\beta}$, then
    \begin{equation*}
        \E \bracks*{\norm*{\mathbf{x}^i(\mathcal{T}) - \mathbf{x}^j(\mathcal{T})}_2^2} \geq {2p \parens*{d -\alpha - \beta}^2}
    \end{equation*}
\end{lemma}

\begin{proof}
    If both events $E(\alpha,S_1,\mathcal{T},i)$ and $E(\beta,S_2,\mathcal{T},j)$ happen in an execution in $\mathcal{T}$, then there are two points $\mathbf{x}_1\in S_1$, $\mathbf{x}_2\in S_2$ such that, $\norm*{\mathbf{x}^i(\mathcal{T}) - \mathbf{x}_1}_2 \leq \alpha$ and $\norm*{\mathbf{x}^j(\mathcal{T}) - \mathbf{x}_2}_2 \leq \beta$.
    Using triangle inequality,
    \begin{align*}
        \norm*{\mathbf{x}_1 - \mathbf{x}_2}_2 &= \norm*{\mathbf{x}_1 - \mathbf{x}^i(\mathcal{T}) + \mathbf{x}^i(\mathcal{T}) - \mathbf{x}^j(\mathcal{T}) + \mathbf{x}^j(\mathcal{T}) - \mathbf{x}_2}_2 \\
        &\leq \norm*{\mathbf{x}_1 - \mathbf{x}^i(\mathcal{T})}_2 + \norm*{\mathbf{x}^i(\mathcal{T}) - \mathbf{x}^j(\mathcal{T})}_2 + \norm*{\mathbf{x}^j(\mathcal{T}) - \mathbf{x}_2}_2
    \end{align*}
    This implies that 
    By rearranging the inequality
    \begin{align*}
        \norm*{\mathbf{x}^i(\mathcal{T}) - \mathbf{x}^j(\mathcal{T})}_2 
        &\geq \norm*{\mathbf{x}_1 - \mathbf{x}_2}_2 - \norm*{\mathbf{x}_1 - \mathbf{x}^i(\mathcal{T})}_2 - \norm*{\mathbf{x}^j(\mathcal{T}) - \mathbf{x}_2}_2\\
        &\geq d - \alpha - \beta.
    \end{align*}
    Similarly, if both events $E(\alpha,S_1,\mathcal{T},j)$ and $E(\beta,S_2,\mathcal{T},i)$ happen in an execution in $\mathcal{T}$, then $ \norm*{\mathbf{x}^j(\mathcal{T}) - \mathbf{x}^i(\mathcal{T})}_2  \geq d - \alpha - \beta$.
    
    Denote the events $E_1 = E(\alpha,S_1,\mathcal{T},i) \cap E(\beta,S_2,\mathcal{T},j)$ and $E_2 = E(\alpha,S_1,\mathcal{T},j) \cap E(\beta,S_2,\mathcal{T},i)$.     
    Assume events $E_1$ and $E_2$ are not disjoint and there is an execution in $\mathcal{T}$ where there are two points $\mathbf{a}\in S_1$ and $\mathbf{b}\in S_2$ such that, $\norm*{\mathbf{x}^i(\mathcal{T}) - \mathbf{a}}_2 \leq \alpha$ and $\norm*{\mathbf{x}^i(\mathcal{T}) - \mathbf{b}}_2 \leq \beta$. This implies that $\norm*{\mathbf{a} - \mathbf{b}}_2 \leq \norm*{\mathbf{x}^i(\mathcal{T}) - \mathbf{a}}_2 + \norm*{\mathbf{x}^i(\mathcal{T}) - \mathbf{b}}_2 = \alpha + \beta < d$, in contradiction to the lemma assumption.
    Since events $E_1$ and $E_2$ must be disjoint, $\Prob \bracks*{E_1 \cup E_2} = \Prob \bracks*{E_1} + \Prob \bracks*{E_2}$. Using the previous inequality,
    \begin{align*}
        \E \bracks*{\norm*{\mathbf{x}^i(\mathcal{T}) - \mathbf{x}^j(\mathcal{T})}_2^2}
        &\geq \parens*{\Prob \bracks*{E_1} + \Prob \bracks*{E_2}} \E \bracks*{\left.\norm*{\mathbf{x}^i(\mathcal{T}) - \mathbf{x}^j(\mathcal{T})}_2^2 \,\right\vert E_1 \cup E_2} \\
        &\geq 2p \parens*{d -\alpha - \beta}^2.
    \end{align*}{}
\end{proof}

To ensure the distributed algorithm implements an SGD algorithm, 
we require it to preserve the same convergence distribution as 
the sequential algorithm, as explained next. 

We say that $\mathcal{A}$ \emph{preserves the convergence distribution
over an execution tree $\mathcal{T}$}, if for any $\beta>0$,
set of points $S$ and process $i$,
\begin{equation}
\label{same-seq-dist}
    \Prob \bracks*{E(\beta,S,\mathcal{T},i)} \geq \Prob\bracks*{E_{seq}(\beta,S)}.
\end{equation}
We say that $\mathcal{A}$ \emph{distributively implements
$\mathcal{A}^{seq}(Q,G,\mathbf{x}_0,\eta,T)$ over an execution tree $\mathcal{T}$} 
if it converge externally \eqref{convergence-require} and internally 
\eqref{diff-require} over the outputs $\braces*{\mathbf{x}^i(\mathcal{T})}_{i=1}^n$, 
and it preserves the convergence distribution over $\mathcal{T}$.

\bigskip
The \emph{$i$-local execution tree} $\mathcal{T}_i$ has a node for 
each of $i$'s local states that appear in execution tree $\mathcal{T}$.
The node is labeled with the probability of $i$ reaching  
this state in $\mathcal{T}$, which is equal to the probability to reach 
some node in $\mathcal{T}$ that contains this local state. 
There is an edge from node $u$ to $v$, if there is an execution 
in $\mathcal{T}$ where the local state in $v$ directly 
follows the local state in $u$. 
This can be extended to an \emph{$S$-local execution tree},
for a set of processes $S$, where each node contains the local states of all processes in $S$.
For simplicity, we assume that once a process transitioned to some 
local state it will never repeat it in the execution, 
e.g., by including the full local history in the local state,
as done in~\cite{goren2021probabilistic}. 

\begin{definition}[Probabilistic Indistinguishably~\cite{goren2021probabilistic}]
Two execution trees are \emph{probabilistically indistinguishable} to 
process $i$, if they induce the same $i$-local execution tree. 
\end{definition}

It was shown~\cite{goren2021probabilistic} that if two execution trees are 
probabilistically indistinguishable to process $i$, then the probability 
for $i$ to perform some action in these two trees is equal.
This easily extends to sets of processes.

\medskip
The proof of the next theorem adapts the impossibility proofs 
of~\cite{AttiyaKS20,HadzilacosHT19} to probabilistic algorithms 
by using probabilistic indistinguishability~\cite{goren2021probabilistic}.
The result is proved for a \emph{weak} adversary, 
which cannot observe the local coin flips, 
shared memory states and the messages sent during the execution. 

\begin{theorem}
\label{thm:thmlb}
If a function $Q$ is ($\gamma$,$p$)-split with 
$p > \frac{\epsmodel}{2\gamma^2}$ and $\gamma > \sqrt{\epsmodel/2}$,
then no algorithm $\mathcal{A}$ distributively implements $\mathcal{A}^{seq}(Q,G,\mathbf{x}_0,\eta,T)$ 
over all valid execution trees which have at most $f$ failures such that $f > \fopt$.
\end{theorem}

\begin{proof}
    Assume there is an algorithm $\mathcal{A}$ that distributively implements $\mathcal{A}^{seq}(Q,G,\mathbf{x}_0,\eta,T)$ over any execution tree $\mathcal{T}$ that has at most $f> \fopt$ failures.
    Since $f> \fopt$, following Lemma~\ref{lemm:fopt-comm}, there are two sets of processes $A$ and $B$, each of size $n-f$, such that $\cluster(A) \cap \cluster(B) = \emptyset$.
    
    Consider the following three execution trees.
    In execution tree $\mathcal{T}_1$, processes in $B$ don't take any local steps, i.e., processes in $B$ are crashed in all executions in $\mathcal{T}_1$.
    Similarly, in execution tree $\mathcal{T}_2$ processes in $A$ don't take any local steps.
    Following our assumption, algorithm $\mathcal{A}$ distributively implements $\mathcal{A}^{seq}(Q,G,\mathbf{x}_0,\eta, T)$ over both $\mathcal{T}_1$ and $\mathcal{T}_2$.
    
    Since $Q$ is $(\gamma,p)$-split, 
    there are two sets of points $S_1, S_2\subseteq \R^d$ 
    where $\mathrm{dist}\parens*{S_1,S_2} \geq d$ 
    and for some $\alpha,\beta\geq 0$ such that $\parens*{d - \alpha - \beta} \geq \gamma$, 
    $\,\Prob \bracks*{E_{seq}(\alpha,S_1)} \Prob \bracks*{E_{seq}(\beta,S_2)} \geq p$.
    
    The third execution tree $\mathcal{T}_3$ is an interleaving of 
    $\mathcal{T}_1$ and $\mathcal{T}_2$, constructed as follows:
    all messages between processes in $A$ and $B$ are delayed 
    until all processes finish executing the algorithm. 
    We build the tree such that the even levels are nodes from $\mathcal{T}_1$, 
    and the odd levels are nodes from $\mathcal{T}_2$. 
    Specifically, nodes in level $2l$ of $\mathcal{T}_3$ 
    are the nodes from the $l$-th level of $\mathcal{T}_1$, 
    and nodes in level $2l+1$ of $\mathcal{T}_3$ 
    are the nodes from the $l$-th level of $\mathcal{T}_2$.
    Since $\cluster(A) \cap \cluster(B) = \emptyset$, 
    processes in $A$ do not share memory with processes in $B$, and vice versa. 
    By construction, processes in one set do not receive messages from the other 
    set until the algorithm terminates.
    Therefore, processes in one set do no affect the processes in the other set,
    and this construction yields a legal execution tree.
    
    By construction, $\mathcal{T}_1$ and $\mathcal{T}_3$ 
    are probabilistically indistinguishable to processes in $A$, and 
    $\mathcal{T}_2$ and $\mathcal{T}_3$ are probabilistically indistinguishable 
   to processes in $B$.     
   In addition, the events $E(\alpha,S_1,\mathcal{T}_3,i)$ and $E(\beta,S_2,\mathcal{T}_3,j)$ 
   are independent of each other for every $i\in A$ and $j\in B$. 
   Therefore, by \eqref{same-seq-dist}, 
    \begin{multline*}
    \Prob \bracks*{E(\alpha,S_1,\mathcal{T}_3,i) \cap E(\beta,S_2,\mathcal{T}_3,j)}
    \\= \Prob \bracks*{E(\alpha,S_1,\mathcal{T}_3,i)} \Prob\bracks*{E(\beta,S_2,\mathcal{T}_3,j)} 
        \geq \Prob \bracks*{E_{seq}(\alpha,S_1)} \Prob\bracks*{E_{seq}(\beta,S_2)} 
        \geq p .
    \end{multline*}
    Similarly, $\Prob \bracks*{E(\alpha,S_1,\mathcal{T}_3,j) \cap E(\beta,S_2,\mathcal{T}_3,i)} \geq p$,
    By Lemma~\ref{impossibility-helper},
    \begin{equation*}
    \\E \bracks*{\norm*{\mathbf{x}^i(\mathcal{T}_3) - \mathbf{x}^j(\mathcal{T}_3)}_2^2} \geq 2{p \parens*{d -\alpha - \beta}^2} \geq 2{p \gamma^2} > \delta.
    \end{equation*}
    This contradicts internal convergence property \eqref{diff-require}
    of Algorithm $\mathcal{A}$.
    This contradicts internal convergence property \eqref{diff-require}
    of Algorithm $\mathcal{A}$.
\end{proof}

\section{Summary}
\label{sec:summary}

We present crash-tolerant asynchronous SGD algorithms for 
the cluster-based model. 
For strongly convex functions, our algorithm obtains maximal speedup
of the convergence rate over the sequential algorithm, 
and tolerates any number of failures.
For other functions, we employ multidimensional approximate agreement 
to bring parameters close together in each iteration.
This algorithm requires that the set of non-failed processes represents 
a majority of the processes. 
We prove that this condition is necessary for optimizing certain functions.

Our results assume processes fail only by crashing, 
which is an adequate model for several computing systems. 
Moreover, concentrating on crash failures allows to obtain good bounds on 
the convergence rate, 
as well as optimal bounds on the ratio of faulty processes. 
This also leads to simpler and more modular proofs.
We believe that the cluster-based model with crash failures 
offers a blueprint for optimization algorithms 
for \emph{high-performance computing} systems.
These architectures include many multi-processor computers, 
each running multiple threads that share a memory space, 
which are connected by a network. 

Future work could explore the use of specific properties of 
objective functions, beyond strong convexity, to improve the algorithms.
Another direction is to have more cooperative inter-cluster computation, e.g., using fetch\&add.



\bibliographystyle{plainurl}
\bibliography{citations}

\appendix

\section{Additional Propositions and Proofs}

\subsection{Useful Mathematical Propositions}
\label{app:use-props}

For every set of vectors $\mathbf{v}_1,\dots,\mathbf{v}_n \in \R^d$,
we use the two following relaxed versions of the triangle inequality~\cite{karimireddy2021scaffold,CollaborativeLearningNeurIPS}. 
\begin{equation}
    \label{relaxed-triangle-ineq2}
    \norm*{\sum_{i=1}^n \mathbf{v}_i}_2^2 \leq n \sum_{i=1}^n \norm*{\mathbf{v}_i}_2^2
\end{equation}
\begin{equation}
    \label{relaxed-triangle-ineq1}
    \forall a>0,\, \norm*{\mathbf{v}_i + \mathbf{v}_j}_2^2 \leq \parens*{1 + a}\norm*{\mathbf{v}_i}_2^2 + \parens*{1 + \frac{1}{a}}\norm*{\mathbf{v}_j}_2^2
\end{equation}

\begin{proposition}
\label{prop:sep-mean-var}
Let $\mathbf{x}\in \R^d$ be a random variable, then
\begin{equation*}
    \E \bracks*{\norm*{\mathbf{x}}_2^2} = \norm*{\E \bracks*{\mathbf{x}}}_2^2 + \Var\bracks*{\mathbf{x}}
\end{equation*}
\end{proposition}

\begin{proof}
    For any random variable $X$, $\parens*{\E\bracks*{X - \E \bracks*{X}}}^2 = \E \bracks*{X^2} - \E\bracks*{X}^2$. The lemma directly follows from this equality.
\end{proof}

\begin{proposition}
\label{prop:var-diff}
Let $\mathbf{x},\mathbf{y}\in \R^d$ be two random variables, then
\begin{equation*}
    \Var \bracks*{\mathbf{x} - \mathbf{y}} \leq 2\Var \bracks*{\mathbf{x}} + 2\Var \bracks*{\mathbf{y}}
\end{equation*}
\end{proposition}

\begin{proof}
    \begin{align*}
        \Var \bracks*{\mathbf{x} - \mathbf{y}} 
        &= \E \bracks*{\norm*{\mathbf{x} - \mathbf{y} - \E \bracks*{\mathbf{x}} + \E \bracks*{\mathbf{y}}}_2^2} \\
        &\leq 2\E \bracks*{\norm*{\mathbf{x} - \E \bracks*{\mathbf{x}}}_2^2} + 2 \E \bracks*{\norm*{\E \bracks*{\mathbf{y}} - \mathbf{y}}_2^2} && \explain{by \eqref{relaxed-triangle-ineq2}} \\
        &= 2\Var \bracks*{\mathbf{x}} + 2\Var \bracks*{\mathbf{y}}
    \end{align*}
\end{proof}

\begin{proposition}
\label{prop:l2-avg}
    Let $S_1,S_2 \subseteq \R^d$, where $|S_1|=s_1$ and $|S_2|=s_2$, then
    \begin{equation*}
        \norm*{\frac{1}{s_1} \sum_{\mathbf{y}\in S_1} \mathbf{y} - \frac{1}{s_2} \sum_{\mathbf{y}'\in S_2} \mathbf{y}'}_2^2 
        \leq \frac{1}{s_1 s_2} \sum_{\mathbf{y}\in S_1} \sum_{\mathbf{y}'\in S_2} \norm*{\mathbf{y} - \mathbf{y}'}_2^2
    \end{equation*}
\end{proposition}

\begin{proof}
    \begin{align*}
    \norm*{\frac{1}{s_1} \sum_{\mathbf{y}\in S_1} \mathbf{y} - \frac{1}{s_2} \sum_{\mathbf{y}'\in S_2} \mathbf{y}'}_2^2 
    &= \norm*{\frac{1}{s_1} \sum_{\mathbf{y}\in S_1} \mathbf{y} - \frac{1}{s_1} \sum_{\mathbf{y}'\in S_2} \frac{s_1}{s_2} \mathbf{y}'}_2^2 \\
    &= \norm*{\frac{1}{s_1} \sum_{\mathbf{y}\in S_1} \mathbf{y} - \frac{1}{s_1} \sum_{\mathbf{y}\in S_1} \frac{1}{s_1} \parens*{\sum_{\mathbf{y}'\in S_2} \frac{s_1}{s_2} \mathbf{y}'}}_2^2 \\
    &\leq \frac{1}{s_1} \sum_{\mathbf{y}\in S_1} \norm*{\mathbf{y} - \frac{1}{s_2} \sum_{\mathbf{y}'\in S_2} \mathbf{y}'}_2^2 && \explain{by \eqref{relaxed-triangle-ineq2}}\\
    &\leq \frac{1}{s_1 s_2} \sum_{\mathbf{y}\in S_1} \sum_{\mathbf{y}'\in S_2} \norm*{\mathbf{y} - \mathbf{y}'}_2^2 && \explain{by \eqref{relaxed-triangle-ineq2}}
\end{align*}
\end{proof}

We also need a weighted version of this proposition.
\begin{proposition}
\label{prop:sep-weights}
    Let $\mathbf{x}_1,\dots,\mathbf{x}_k, \mathbf{y}_1, \dots, \mathbf{y}_k \in \R^d$ and weights $\sum_{i=1}^k w_i = 1$, then
    \begin{equation*}
        \norm*{\sum_{i=1}^k w_i \mathbf{x}_i - \sum_{i=1}^k w_i \mathbf{y}_k}_2^2 \leq \sum_{i=1}^{k} w_i \norm*{\mathbf{x}_i - \mathbf{y}_i}_2^2 .
    \end{equation*}
    This also implies for any $\mathbf{y}\in \R^d$ that
    \begin{equation*}
        \norm*{\sum_{i=1}^k w_i \mathbf{x}_i - \mathbf{y}}_2^2 \leq \sum_{i=1}^{k} w_i \norm*{\mathbf{x}_i - \mathbf{y}}_2^2 
    \end{equation*}
\end{proposition}

\begin{proof}
Define $a_i = \frac{1 - \sum_{j=1}^{i} w_j}{w_i}$. Note that $1 + a_i = \frac{1 - \sum_{j=1}^{i-1} w_j}{w_i}$ and $1+\frac{1}{a_i} = \frac{1 - \sum_{j=1}^{i-1} w_j}{1 - \sum_{j=1}^{i} w_j}$. By repeatedly applying \eqref{relaxed-triangle-ineq1} we get that,
    \begin{multline*}
        \norm*{\sum_{i=1}^k w_i \mathbf{x}_i - \sum_{i=1}^k w_i \mathbf{y}_i}_2^2 \\
        \leq \sum_{i=1}^{k-1} \parens*{1 + a_i} \prod_{l=1}^{i-1} \parens*{1 + \frac{1}{a_l}}\norm*{w_i \mathbf{x}_i - w_i \mathbf{y}_i}_2^2  + \prod_{l=1}^{k-1} \parens*{1 + \frac{1}{a_l}} \norm*{w_k \mathbf{x}_k - w_k \mathbf{y}_i}_2^2 \\
        = \sum_{i=1}^{k-1} \frac{1 - \sum_{j=1}^{i-1} w_j}{w_i} \prod_{l=1}^{i-1} \frac{1 - \sum_{j=1}^{l-1} w_j}{1 - \sum_{j=1}^{l} w_j} \norm*{w_i \mathbf{x}_i - w_i \mathbf{y}_i}_2^2 + \prod_{l=1}^{k-1} \frac{1 - \sum_{j=1}^{l-1} w_j}{1 - \sum_{j=1}^{l} w_j} \norm*{w_k \mathbf{x}_k - w_k \mathbf{y}_i}_2^2 
    \end{multline*}
Since $\prod\limits_{l=1}^{i-1} \frac{1 - \sum_{j=1}^{l-1} w_j}{1 - \sum_{j=1}^{l} w_j} = \frac{1}{1 - \sum_{j=1}^{i-1} w_j}$ and $1 - \sum_{j=1}^{k-1} w_j = w_k$, we get that 
\begin{equation*}
        \norm*{\sum_{i=1}^k w_i \mathbf{x}_i - \sum_{i=1}^k w_i \mathbf{y}_i}_2^2 
        \leq \sum_{i=1}^{k} \frac{1}{w_i} \norm*{w_i \mathbf{x}_i - w_i \mathbf{y}_i}_2^2 
        = \sum_{i=1}^{k} w_i \norm*{\mathbf{x}_i - \mathbf{y}_i}_2^2 
    \end{equation*}
\end{proof}

We sometimes use the next property of the inner product 
\begin{equation}
\label{inned-prod-add}
    \tup{\mathbf{x} + \mathbf{y},\mathbf{z}} = \tup{\mathbf{x},\mathbf{z}} + \tup{\mathbf{y},\mathbf{z}}
\end{equation}
In addition, \emph{Young’s inequality} states that for every $\beta>0$
\begin{equation}
\label{eq:young}
\begin{aligned}
     |\tup{\mathbf{x}, \mathbf{y}}| \leq \frac{\beta}{2}\norm*{\mathbf{x}}_2^2 + \frac{1}{2\beta}\norm*{\mathbf{y}}_2^2
\end{aligned}
\end{equation}

\subsection{Additional Proofs}
\label{app:add-proofs}

\vartotal*

\begin{proof}
For every $1\leq i \leq b$, denote $\mathbf{g}_i = G(\mathbf{x}_i,z_i)$.
\begin{align*}
    \Var\bracks*{\frac{1}{b}\sum_{i=1}^b \mathbf{g}_i}
    &= \E\bracks*{\norm*{\frac{1}{b}\sum_{i=1}^b \mathbf{g}_i - \E\bracks*{\frac{1}{b}\sum_{i=1}^b \mathbf{g}_i}}_2^2} \\
    &= \E\bracks*{\norm*{\frac{1}{b}\sum_{i=1}^b \parens*{\mathbf{g}_i - \E\bracks*{\mathbf{g}_i}}}_2^2} && \explain{Linearity of expectation} \\
    &= \frac{1}{b^2} \E \bracks*{\tup*{\sum_{i=1}^b \parens*{\mathbf{g}_i - \E\bracks*{\mathbf{g}_i}}, \sum_{i=1}^b \parens*{\mathbf{g}_i - \E\bracks*{\mathbf{g}_i}}}}\\
    &= \frac{1}{b^2} \sum_{i=1}^b \sum_{j=1}^b \E \bracks*{\tup*{\mathbf{g}_i - \E \bracks*{\mathbf{g}_i} ,\mathbf{g}_j - \E\bracks*{\mathbf{g}_j}}} && \explain{by \eqref{inned-prod-add}}
\end{align*}
Note that $z_i$ and $z_j$ are independent for $i\neq j$, and therefore, $\mathbf{g}_i$ and $\mathbf{g}_j$ are also stochastically independent and we get that
\begin{equation*}
    \E \bracks*{\tup*{\mathbf{g}_i - \E \bracks*{\mathbf{g}_i} ,\mathbf{g}_j - \E\bracks*{\mathbf{g}_j}}}
    = \tup*{\E \bracks*{\mathbf{g}_i - \E \bracks*{\mathbf{g}_i}} ,\E \bracks*{\mathbf{g}_i - \E\bracks*{\mathbf{g}_j}}} = 0 
\end{equation*}
Using this in the first equation yields
\begin{align*}
    \Var\bracks*{\frac{1}{b}\sum_{i=1}^b \mathbf{g}_i}
    &= \frac{1}{b^2} \sum_{i=1}^b \sum_{j=1}^b \E \bracks*{\tup*{\mathbf{g}_i - \E \bracks*{\mathbf{g}_i} ,\mathbf{g}_j - \E\bracks*{\mathbf{g}_j}}} \\
    &= \frac{1}{b^2} \sum_{i=1}^b \E \bracks*{\norm*{\mathbf{g}_i - \E \bracks*{\mathbf{g}_i}}_2^2} \\
    &= \frac{1}{b^2} \sum_{i=1}^b \Var \bracks*{\mathbf{g}_i} \leq \frac{\sigma^2}{b} && \explain{by \eqref{bounded-var}}
\end{align*}
\end{proof}

\bottouoneovert*

\begin{proof}
    The proof is by induction. For the base case, 
    \begin{align*}
        a_2 &\leq \parens*{1 - \frac{\beta\mu}{(\gamma + 1)}} a_1 + \frac{\beta^2 c}{b(\gamma + 1)^2}\\
        &\leq \parens*{1 - \frac{\beta\mu}{b(\gamma + 1)}} a_1 + \frac{\beta^2 c}{b(\gamma + 1)^2} && \explain{$b\geq 1$}\\
        &= \parens*{\frac{\gamma}{b\parens*{\gamma + 1}^2} - \frac{\beta\mu - 1}{b\parens*{\gamma + 1}^2}} \parens*{\gamma + 1} a_1 + \frac{\beta\mu -1}{b(\gamma + 1)^2} \frac{\beta^2 c}{\beta\mu - 1} \\
        &\leq \parens*{\frac{\gamma}{b\parens*{\gamma + 1}^2} - \frac{\beta\mu - 1}{b\parens*{\gamma + 1}^2}} \nu + \frac{\beta\mu -1}{b(\gamma + 1)^2} \nu \\
        &= \frac{\gamma}{\parens*{\gamma + 1}^2} \frac{\nu}{b} \leq \frac{\nu}{\parens*{\gamma + 2}b}
    \end{align*}
    
    For the induction step,
    \begin{align*}
    a_{t+1} &\leq \parens*{1 - \eta_t \mu} a_t + \frac{c\eta_t^2}{b} \\
    &\leq \parens*{1 - \frac{\beta\mu}{\gamma + t}}\frac{\nu}{b\parens*{\gamma + t}} + \frac{\beta^2 c}{b(\gamma + t)^2}\\
    &= \parens*{\frac{\gamma + t - \beta\mu}{b\parens*{\gamma + t}^2}}\nu + \frac{\beta^2 c}{b(\gamma + t)^2} \\
    &\leq \parens*{\frac{\gamma + t - 1}{b\parens*{\gamma + t}^2}}\nu - \underbrace{\parens*{\frac{\beta\mu - 1}{b\parens*{\gamma + t}^2}}\nu + \frac{\beta^2 c}{b(\gamma + t)^2}}_{\text{$\leq 0$ by the definition of $\nu$}}\leq \frac{\nu}{b\parens*{\gamma + t + 1}}
\end{align*}
\end{proof}

We restate Lemma~\ref{small-diam} to also account for the difference between the learning parameters of the first iteration. Since $\max_{i,j\in V_0} \norm*{\mathbf{x}_{1}^i - \mathbf{x}_{1}^j}_2^2 = 0$ in Algorithm~\ref{alg:dis-sgd}, we get the lemma as stated in Section~\ref{section:non convex}.

\begin{lemma}
\label{app-small-diam}
    Consider Algorithm~\ref{alg:dis-sgd} with constant learning rate 
    $\eta \leq \min \braces*{\frac{1}{2},\frac{1}{L}}$ and
    $q = \frac{\eta}{4}$. Then for every iteration $t\geq 1$,
    \begin{equation*}
        \max_{i,j\in V_t} \E\bracks*{\norm*{\mathbf{x}_{t+1}^i - \mathbf{x}_{t+1}^i}_2^2} \leq  \frac{2\sigma^2 \eta^3}{\thresh}
        + \frac{1}{2^t} \max_{i,j\in V_0} \norm*{\mathbf{x}_{1}^i - \mathbf{x}_{1}^j}_2^2
    \end{equation*}
\end{lemma}

\begin{proof}
Following Lemma~\ref{rec-model-diam}
\begin{align*}
    \max_{i,j\in V_t} \E\bracks*{\norm*{\mathbf{x}_{t+1}^i - \mathbf{x}_{t+1}^i}_2^2} 
    &\leq  q \parens*{2 + 2L^2\eta^2} \max_{i,j\in V_{t-1}}\E\bracks*{\norm*{\mathbf{x}_t^i - \mathbf{x}_t^j}_2^2} + q \frac{4\sigma^2\eta^2}{\thresh}\\
    &= \frac{1}{2}\parens*{\eta + L^2\eta^3} \max_{i,j\in V_{t-1}}\E\bracks*{\norm*{\mathbf{x}_t^i - \mathbf{x}_t^j}_2^2} + \frac{\sigma^2\eta^3}{\thresh} \\
    &\leq \eta \max_{i,j\in V_{t-1}}\E\bracks*{\norm*{\mathbf{x}_t^i - \mathbf{x}_t^j}_2^2} + \frac{\sigma^2\eta^3}{\thresh} && \explain{$\eta \leq \frac{1}{L}$} \\
    &\leq \frac{1}{2} \max_{i,j\in V_{t-1}}\E\bracks*{\norm*{\mathbf{x}_t^i - \mathbf{x}_t^j}_2^2} + \frac{\sigma^2\eta^3}{\thresh} && \explain{$\eta \leq \frac{1}{2}$} 
    \end{align*}
    
    By denoting $u_t = \max_{i,j\in V_{t-1}}\E\bracks*{\norm*{\mathbf{x}_t^i - \mathbf{x}_t^j}_2^2}$ we can write $u_{t+1} \leq \frac{1}{2} u_t + \frac{\sigma^2 \eta^3}{\thresh}$. 
    By induction 
    \begin{equation*}
        u_{t+1} \leq \frac{\sigma^2 \eta^3}{\thresh} \sum_{i=0}^{t-1} 2^{-i} + 2^{-t} \,u_1
        \leq \frac{\sigma^2 \eta^3}{\thresh} \sum_{i=0}^{\infty} 2^{-i} + 2^{-t} \,u_1
        \leq \frac{2\sigma^2\eta^3 }{\thresh} + 2^{-t} \,u_1
    \end{equation*}{}
\end{proof}


    

\section{Non-Convex Algorithm with Different Initial Points}
\label{sec:differentinitialpoints}

This section assumes that each process $i$ starts at a different initial point $\mathbf{x}_1^i$ in the non-convex algorithm (Algorithm~\ref{alg:dis-sgd}), as considered in the strongly-convex case (Algorithm~\ref{alg:dis-sgd-strongly-convex}). The next lemma uses the same parameters as in Lemma~\ref{lem:non-convex-small-error}, with the assumption of different starting points.

\begin{lemma}
\label{lem:different-initial-convergence-rate}
Let $Q$ be an $L$-smooth cost function. 
Consider Algorithm~\ref{alg:dis-sgd} with different initial points $\braces*{\mathbf{x}_1^i}_{i=1}^n$, $T\geq \max\braces*{16L^2\thresh,4\thresh}$, constant learning rate $\eta = \frac{\sqrt{\thresh}}{\sqrt{T}}$ and parameters set as in Lemma~\ref{small-diam},
then for every process $i\in V_T$
    \begin{equation*}
        \E\bracks*{\norm*{\nabla Q(\mathbf{x}_\tau^i)}_2^2} 
        \leq
        \frac{4 \parens*{Q(\mathbf{x}_1) - Q^*}}{\sqrt{\thresh T}} + \frac{20\sigma^2 \max\braces*{L,L^2}}{\sqrt{\thresh T}} + \frac{16\sigma^2}{\sqrt{\thresh T}}
        + \parens*{\frac{16}{\thresh}
        + \frac{16L^2}{T}} \max_{i,j\in V_0} \norm*{\mathbf{x}_1^i - \mathbf{x}_1^j}_2^2
    \end{equation*}
\end{lemma}

\begin{proof}
    Note that $\eta = \frac{\sqrt{\thresh}}{\sqrt{T}}\leq \min\braces*{\frac{1}{4L},\frac{1}{2}}$.
    By Lemma~\ref{app-small-diam}, 
    \begin{align*}
        \sum_{t=2}^T \max_{i,j\in V_t} \E \bracks*{\norm*{\mathbf{x}_t^i - \mathbf{x}^j_t}_2^2} 
        &\leq \parens*{T-1}\, \frac{2\sigma^2\eta^3}{\thresh{}} + \sum_{t=2}^T \frac{1}{2^t} \max_{i,j\in V_0} \norm*{\mathbf{x}_{1}^i - \mathbf{x}_{1}^i}_2^2 \\
        &\leq  \frac{2\sqrt{\thresh{}}\sigma^2}{\sqrt{T}} + \sum_{t=2}^T \frac{1}{2^t} \max_{i,j\in V_0} \norm*{\mathbf{x}_{1}^i - \mathbf{x}_{1}^i}_2^2
    \end{align*}
    Following Lemma~\ref{thm:non-convex} and the previous inequality
    \begin{align*}
        &\E \bracks*{\norm*{\nabla Q(\mathbf{x}_\tau^i)}_2^2} \\
        &= \frac{4 \parens*{Q(\mathbf{x}_1) - Q^*}}{\sqrt{\thresh T}} + \frac{4\sigma^2 L}{\sqrt{\thresh T}} + \parens*{\frac{8}{\thresh} + \frac{8L^2}{T}} \parens*{\frac{2\sqrt{\thresh}\sigma^2}{\sqrt{T}} + \max_{i,j\in V_0} \norm*{\mathbf{x}_1^i - \mathbf{x}_1^j}_2^2 \sum_{t=0}^{T-1} 2^{-t}} \\
         &= \frac{4 \parens*{Q(\mathbf{x}_1) - Q^*}}{\sqrt{\thresh T}} + \frac{4\sigma^2 L}{\sqrt{\thresh{}T}} + \frac{16\sigma^2}{\sqrt{\thresh{}T}} + \frac{16 \sqrt{\thresh} L^2\sigma^2}{T^{3/2}} + \parens*{\frac{8}{\thresh} + \frac{8L^2}{T}} \max_{i,j\in V_0} \norm*{\mathbf{x}_1^i - \mathbf{x}_1^j}_2^2 \sum_{t=0}^{\infty} 2^{-t} \\
         &\leq \frac{4 \parens*{Q(\mathbf{x}_1) - Q^*}}{\sqrt{\thresh T}} + \frac{20\sigma^2 \max\braces*{L,L^2}}{\sqrt{\thresh{}T}} + \frac{16\sigma^2}{\sqrt{\thresh T}} + \parens*{\frac{16}{\thresh} + \frac{16L^2}{T}} \max_{i,j\in V_0} \norm*{\mathbf{x}_1^i - \mathbf{x}_1^j}_2^2
    \end{align*}
    Where the last inequality uses that $T\geq \thresh$.
\end{proof}

The analysis in Lemma~\ref{lem:different-initial-convergence-rate} yields an error term of $\Omega \parens*{\max_{i,j\in V_0} \norm*{\mathbf{x}_1^i - \mathbf{x}_1^j}_2^2}$. 
To deal with this and obtain external convergence, we can add an initial MDAA round, before starting the algorithm's iterations in Algorithm~\ref{alg:dis-sgd}. Using the previous lemma, we can show that with an additional of $\approx \log T$ communication rounds the algorithm achieves external convergence with different initial points. However, the error term depends on the initial distance between the starting points, which decreases with $T$ at the same rate as the other components of the error term. 

\begin{theorem}
    Let $Q$ be an $L$-smooth cost function. 
    Consider Algorithm~\ref{alg:dis-sgd} with different initial points $\braces*{\mathbf{x}_1^i}_{i=1}^n$ and initial MDAA round satisfying $({1}/{\sqrt{T})}$-contraction, $T\geq \max\braces*{16L^2\thresh,4\thresh}$, constant learning rate $\eta = \frac{\sqrt{\thresh}}{\sqrt{T}}$ and parameters set as in Lemma~\ref{small-diam},
    then for every process $i\in V_T$
    \begin{equation*}
        \E\bracks*{\norm*{\nabla Q(\mathbf{x}_\tau^i)}_2^2} 
        \leq
        \frac{4 \parens*{Q(\mathbf{x}_1) - Q^*}}{\sqrt{\thresh T}} + \frac{20\sigma^2 \max\braces*{L,L^2}}{\sqrt{\thresh T}} + \frac{16\sigma^2}{\sqrt{\thresh T}} + \parens*{\frac{16}{\thresh\sqrt{T}} + \frac{16L^2}{T^{3/2}}} \max_{i,j\in V_0} \norm*{\mathbf{x}_1^i - \mathbf{x}_1^j}_2^2
    \end{equation*}
\end{theorem}

Note that an adaptive contraction parameter can be used to achieve a constant distance between the parameters of the first iteration (i.e., $\epsilon$-agreement), thus, eliminating the dependency on the initial distance between the points (see~\cite{AbrahamAD04}).

The next theorem shows that the analysis in Lemma~\ref{lem:different-initial-convergence-rate} is {tight}, by showing that the error term of $\Omega \parens*{\max_{i,j\in V_0} \norm*{\mathbf{x}_1^i - \mathbf{x}_1^j}_2^2}$ in unavoidable for \emph{any} constant learning rate $\eta$ and contraction parameter $q$. 
This means that when Algorithm~\ref{alg:dis-sgd}, without an initial MDAA round, uses a constant learning rate and contraction parameter with different initial points, it cannot provide external convergence.

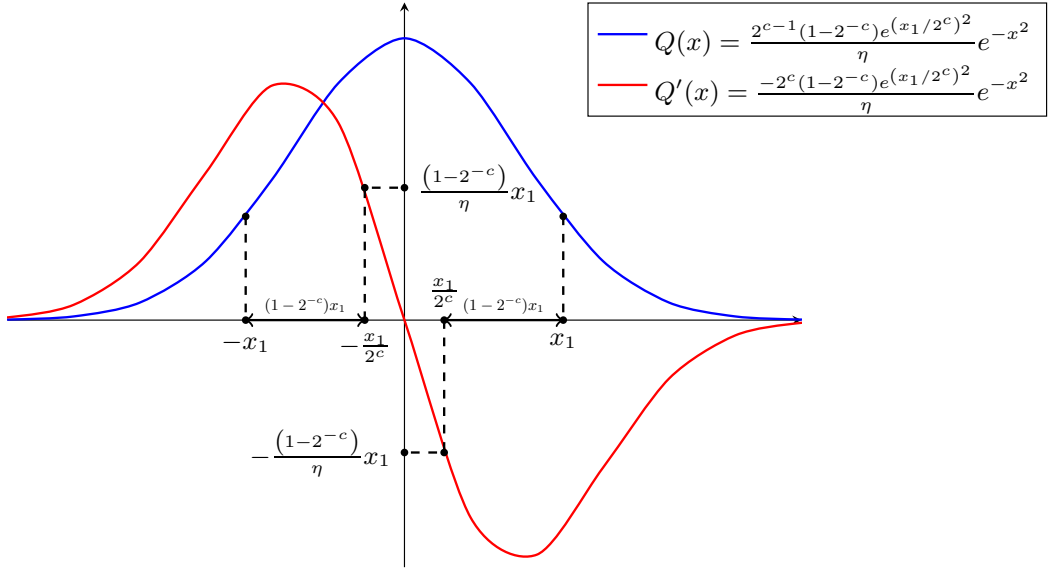
\begin{figure}[t]
\centering
\begin{tikzpicture}[declare function = {f(\x) = 0.5*2*0.75*exp(0.0625)*exp(-\x^2);}, declare function = {g(\x) = -0.5*4*0.75*exp(0.0625)*\x* exp(-\x^2);},scale=1.1]
    \tikzstyle{every node}=[font=\small]
  \begin{axis} [axis lines=center, 
    xmin=-2.5,xmax=2.5,
    ymin=-0.7,ymax=0.9,
    ticks=none,
    x post scale = 1.4,
    y post scale = 1.2,
    legend style={at={(1.02,1.00)},anchor=north,legend cell align=left}]
    \addplot [smooth, thick, blue] {f(x)};
    \addplot [smooth, thick, red] {g(x)};
    
    \node[circle,fill,scale=0.3,label=below:$x_1$] at (axis cs:1,0){};
    \node[circle,fill,scale=0.3,label=below:$-x_1$] at (axis cs:-1,0){};
    \node[circle,fill,scale=0.3] at (axis cs:1,{f(1)}){};
    \node[circle,fill,scale=0.3] at (axis cs:-1,{f(-1)}){};
    \addplot[color = black, dashed, thick] coordinates {(1, {f(1))}) (1, 0)};
    \addplot[color = black, dashed, thick] coordinates {(-1, {f(-1)}) (-1, {0})};
    
    \node[circle,fill,scale=0.3,label=above:$\frac{x_1}{2^c}$] at (axis cs:0.25,0){};
    \node[circle,fill,scale=0.3,label=below:$-\frac{x_1}{2^c}$] at (axis cs:-0.25,0){};
    \node[circle,fill,scale=0.3,label=left:$-\frac{\parens*{1-2^{-c}}}{\eta}x_1$] at (axis cs:0,{g(0.25)}){};
    \node[circle,fill,scale=0.3,label=right:$\frac{\parens*{1-2^{-c}}}{\eta}x_1$] at (axis cs:0,{g(-0.25)}){};
    \node[circle,fill,scale=0.3] at (axis cs:0.25,{g(0.25)}){};
    \node[circle,fill,scale=0.3] at (axis cs:-0.25,{g(-0.25)}){};
    
    \addplot[color = black, dashed, thick] coordinates {(0.25, 0) (0.25, {g(0.25)}) (0,{g(0.25)})};
    \addplot[color = black, dashed, thick] coordinates {(-0.25, 0) (-0.25, {g(-0.25)}) (0, {g(-0.25)})};
    
     \draw[<->, thick] (axis cs:-0.25, {0}) -- (axis cs:-1, {0}) node [above,pos=0.5,scale=0.6] {$(1-2^{-c})x_1$};
     
     \draw[<->, thick] (axis cs:0.25, {0}) -- (axis cs:1, {0}) node [above,pos=0.5,scale=0.6] {$(1-2^{-c})x_1$};;
    
    \addlegendentry{$Q(x) =  \frac{2^{c-1}(1-2^{-c}) e^{\parens*{x_1/2^c}^2}}{\eta} e^{-x^2}$}
    \addlegendentry{$Q'(x) = \frac{-2^{c}(1-2^{-c}) e^{\parens*{x_1/2^c}^2}}{\eta} e^{-x^2}$}
  \end{axis}
\end{tikzpicture}
\caption{Plot of $Q$ and its gradient function used in the proof of Theorem~\ref{thm:different-initial-lb}}
\label{fig:func1}
\end{figure}

\begin{theorem}
\label{thm:different-initial-lb}
Consider Algorithm~\ref{alg:dis-sgd} with constant learning rate $\eta$, such that $\eta_t = \eta$ for every $t\geq 1$, and any contraction parameter $q<1$ for the MDAA algorithm in Line~\ref{lin:approx-agree}.
For any $\mathbf{x}_1 \in \R^d$, there exists an $L$-smooth cost function $Q$ and set of initial parameters $\braces*{\mathbf{x}_1^i}_{i=1}^n$ such that $\max_{i\in V_0} \norm*{\mathbf{x}_1^i}_2 = \norm*{\mathbf{x}_1}_2$
and the output of some process $i$ from Algorithm~\ref{alg:dis-sgd}, for any number for iterations $T$, over some valid execution tree $\mathcal{T}$ has error
\begin{align*}
    \E_{\mathcal{T}}\bracks*{\norm*{\nabla Q(\mathbf{x}^i)}_2^2} 
    = \Omega \parens*{ \max_{i\in V_0} \norm*{\mathbf{x}_1^i}_2^2}
\end{align*}
\end{theorem}

\begin{proof}
Consider a system with singleton clusters, which is simply the pure message-passing model. In this case, $\fopt = \lfloor(n-1)/2\rfloor$ and assume that the number of failures is exactly $f=\fopt$. We consider a cost function $Q$ of dimension $d=1$, hence, we can use the Midpoint aggregation rule in the approximate agreement algorithm, which achieves contraction rate of 1/2.
For simplicity, assume that $\thresh{} = 3$ and note that the proof can be generalized for $\thresh{}>3$.

Let integer $c$ such that $\frac{1}{2^c}\leq q$.
Let $x_1\in \R$ and consider the function $Q(x) = \frac{2^{c-1}(1-2^{-c}) e^{\parens*{x_1/2^c}^2}}{\eta} e^{-x^2}$. The function $Q$ is $L$-smooth (for some $L$) since its second derivative is bounded.
Its derivative is $\nabla Q(x) = Q'(x) = \frac{-2^{c}(1-2^{-c}) e^{\parens*{x_1/2^c}^2}}{\eta} x e^{-x^2}$.
The function $Q(x)$ and its derivative are depicted in Figure~\ref{fig:func1}.
The function $Q'$ is \emph{odd}, since $-Q'(x) = Q'(-x)$.
The stochastic gradient is simply set to be the full gradient. 

Assume that $n\geq 2\thresh{} + 1$ and is an odd integer. 
Divide the processes, except process $1$, into two sets $S_1$ and $S_2$, each containing $\lfloor n/2\rfloor$ processes. (If $n$ is even, one process can crash at the beginning of the execution.) All the processes $i\in S_1$ start with ${x}_1^i = x_1$, all the processes $i\in S_2$ start with ${x}_1^i = -x_1$ and process $1$ starts with $x_1^1 = 0$.
As the stochastic gradient is deterministic, the execution tree reduces to a simple deterministic execution.
Consider the following execution of Algorithm~\ref{alg:dis-sgd}; at the first iteration all the processes in $S_1$ send $Q'(x_1)$, processes in $S_2$ send $Q'(-x_1)$ and process $1$ sends $Q'(0) = 0$. Each process receives exactly $\thresh = 3$ gradients, one of each kind.
As the gradient is odd, for every processes $i$, ${g}_t^i = \frac{1}{3} \parens*{0 + Q'(x) + Q'(-x)} = 0$. Hence, each process $i$ updates ${y}_1^i$ to be the same as its initial point. In the AA algorithm of the first iteration, at each round, the processes in $S_1$ receive messages only from processes in $S_1$ and process $1$. Similarly, the processes in $S_2$ receive messages only from processes in $S_2$ and process $1$. Process $1$ receives all the messages from both processes in $S_1$ and $S_2$. Hence, each process $i\in S_1$ updates its parameter to be $x_2^i = x_1/2^c$, each process $i\in S_2$ updates its parameter to be $x_2^i = -x_1/2^c$ and process $1$ stays with $x_2^1 = 0$. 

Consider the second iteration, where each process in $S_1$ and $S_2$ receives gradients only from its own group, and process $1$ receives one of each kind. Since $\nabla Q(x_1/2^c) = \frac{-(1-2^{-c})}{\eta} x_1$, each process $i\in S_1$ updates its parameter to be $y_2^i = x_1$, each process $i\in S_2$ updates its parameters to be $y_2^i = -x_1$ and process $1$ stays with $y_2^1 = 0$. The AA phase of the second iteration carries out exactly like in the first one. The rest of the iterations carry on exactly as the second iteration. We get that for the processes $i\in S_1\cup S_2$ and iterations $2\leq t \leq T+1$, $\abs*{\nabla Q(x_t^i)} = \abs*{\frac{(1-2^{-c})}{\eta} x_1}$, which concludes the lemma.
\end{proof}

Since $\max_{i,j\in V_0} \norm*{\mathbf{x}_1^i - \mathbf{x}_1^j}_2 \leq 2 \max_{i\in V_0} \norm*{\mathbf{x}_1^i}_2$, \autoref{thm:different-initial-lb} implies an error of $\Omega \parens*{ \max_{i\in V_0} \norm*{\mathbf{x}_1^i}_2^2} = \Omega \parens*{\max_{i,j\in V_0} \norm*{\mathbf{x}_1^i - \mathbf{x}_1^j}_2^2}$.

\end{document}